 \documentclass[a4paper]{article}

\usepackage[boxruled,lined,linesnumbered]{algorithm2e}
\usepackage{amsthm,graphicx,graphics,amsmath,amssymb,enumitem}
\usepackage{comment}
\usepackage{caption,url}

\title{Depth First Search in the Semi-streaming Model}

\author{Shahbaz Khan\thanks{This research work was supported by the European Research Council under the European Union's Seventh Framework Programme (FP/2007-2013)~/~ERC Grant Agreement no. 340506.}\\Faculty of Computer Science\\University of Vienna, Austria\\ \texttt{shahbaz.khan@univie.ac.at} \and
Shashank K. Mehta\\Dept. of Computer Science and Engineering\\Indian Institute of Technology Kanpur, India\\ \texttt{skmehta@cse.iitk.ac.in}}

\newtheorem{theorem}{Theorem}[section]
\newtheorem{lemma}{Lemma}[section]
\newtheorem{remark}{Remark}[section]
\newtheorem{observation}{Observation}[]

\newcommand{\II}{{\cal I}}

\date{}
\begin{document}
\maketitle

\begin{abstract}
Depth first search (DFS) tree is a fundamental data structure for solving various graph problems. The classical algorithm for building a DFS tree requires $O(m+n)$ time for a given undirected graph $G$ having $n$ vertices and $m$ edges.  In the streaming model, an algorithm is allowed several passes (preferably single) over the input graph having a restriction on the size of local space used. 

Now, a DFS tree of a graph can be trivially computed using a single pass if $O(m)$ space is allowed. In the semi-streaming model allowing $O(n)$ space, it can be computed in $O(n)$ passes over the input stream, where each pass adds one vertex to the DFS tree. However, it remains an open problem to compute a DFS tree using $o(n)$ passes  using $o(m)$ space even in any relaxed streaming environment.

We present the first semi-streaming algorithms that compute a DFS tree of an undirected graph in $o(n)$ passes using $o(m)$ space. We first describe an extremely simple algorithm that requires at most $\lceil n/k\rceil$ passes to compute a DFS tree using $O(nk)$ space, where $k$ is any positive integer. For example using $k=\sqrt{n}$, we can compute a DFS tree in $\sqrt{n}$ passes using $O(n\sqrt{n})$ space. We then improve this algorithm by using more involved techniques to reduce the number of passes to  $\lceil h/k\rceil$ under similar space constraints, where $h$ is the height of the computed DFS tree. In particular, this algorithm improves the bounds for the case where the computed DFS tree is {\em shallow} (having $o(n)$ height). Moreover, this algorithm is presented in form of a {\em framework} that allows the {\em flexibility} of using any algorithm to maintain a DFS tree of a stored sparser subgraph as a {\em black box}, which may be of an independent interest. Both these algorithms essentially demonstrate the existence of a trade-off between the space and number of passes required for computing a DFS tree. Furthermore, we evaluate these algorithms experimentally which reveals their exceptional performance in practice. For both random and real graphs, they require merely a {\em few} passes even when allowed just $O(n)$ space. \\


\end{abstract}

\newpage

\section{Introduction}
Depth first search (DFS) is a well known graph traversal technique. 
Right from the seminal work of Tarjan~\cite{Tarjan72}, DFS traversal 
has played an important role in the design of efficient algorithms for many 
fundamental graph problems, namely, bi-connected components, 
strongly connected components, topological sorting~\cite{Tarjan76}, 
dominators in directed graph~\cite{Tarjan74}, etc. Even in undirected graphs, 
DFS traversal have various applications including computing connected components, cycle detection, 
edge and vertex connectivity~\cite{EvenT75} (via articulation points and bridges), 
bipartite matching~\cite{HopcroftK73}, planarity testing~\cite{HopcroftT74} etc. 
In this paper, we address the problem of computing a DFS tree in the semi-streaming environment. 


The streaming model~\cite{AlonMS99,FeigenbaumKSV03,GuhaKS01} is a popular model for computation 
on large data sets wherein a lot of algorithms have been developed~\cite{FlajoletM85,HenzingerRR98,GuhaKS01,Indyk06} to address significant problems in this model.  
The model requires the entire input data to be accessed as a stream, typically in a single pass
over the input, allowing very small amount of storage ($poly\log$ in input size).
A streaming algorithm 
must judiciously choose the data to be saved in the small space,
so that the computation can be completed successfully. In the context of graph problems, 
this model is adopted in the following fashion. For a given graph $G=(V,E)$ having $n$ vertices, 
an input stream sends the graph edges in $E$ using an arbitrary order only once, and the 
allowed size of local storage is $O(poly\log n)$. The algorithm iteratively asks for the next edge and 
performs some computation. After the stream is over, the final computation is performed and the 
result is reported. At no time during the entire process should the total size of stored data 
exceed $O(poly\log n)$.

In general only statistical properties of the graph are computable under this 
model, making it impractical for use in more complicated graph problems~\cite{FeigenbaumKMSZ05b,GuruswamiK16}.
A prominent exception for the above claim is the problem of counting triangles 
($3$-cycles) in a graph~\cite{Bar-YossefKS02}. 
Consequently, several relaxed models have been proposed with a goal to solve more complex 
graph problems. One such model is called {\em semi-streaming model}~\cite{Muthukrishnan05,FeigenbaumKMSZ05} which 
relaxes the storage size to $O(n~poly\log n)$. Several significant problems have been studied under this model (surveys in~\cite{ConnellC09,Zhang10,McGregor14}). Moreover, even though it is preferred to allow only a single pass over the input stream, several hardness results~\cite{HenzingerRR98,BuchsbaumGW03,FeigenbaumKMSZ05,BorradaileMM14,GuruswamiO16} have reported the limitations of using a single pass (or even $O(1)$ passes). This has led to the development of various multi-pass algorithms~\cite{FeigenbaumKMSZ05,FeigenbaumKMSZ05b,McGregor05,AhnG13,Kapralov13,KaleT17} in this model.
Further, several streaming algorithms maintaining approximate distances~\cite{FeigenbaumKMSZ05b,Baswana08,Elkin11} are also
known to require $O(n^{1+\epsilon})$ space (for some constant $\epsilon>0$) relaxing the requirement of $O(n~poly\log n)$ space.


Now, a DFS tree of a graph can be computed in a single pass if $O(m)$ space is allowed. 
If the space is restricted to $O(n)$, it can be trivially computed using $O(n)$ passes over the input stream, 
where each pass adds one vertex to the tree. This can also be easily improved to $O(h)$ passes, 
where $h$ is the height of the computed DFS tree. Despite most applications of DFS trees in undirected graphs 
being efficiently solved in the semi-streaming environment~\cite{WestbrookT92,FeigenbaumKMSZ05,FeigenbaumKMSZ05b,AusielloFL09,AusielloFL12,Farach-ColtonHL15,Kliemann16}, due to its fundamental nature DFS is considered a long standing open problem~\cite{Farach-ColtonHL15,ConnellC09,Ruhl03} even for undirected graphs. Moreover, computing a DFS tree in ${O}(poly\log n)$ passes is considered hard~\cite{Farach-ColtonHL15}. To the best of our knowledge, it remains an open problem to compute a DFS tree using $o(n)$ passes even in any relaxed streaming environment. 


In our results, we borrow some key ideas from recent sequential algorithms~\cite{BaswanaK17,BaswanaCCK16} 
for maintaining dynamic DFS of undirected graphs. 
Recently, similar ideas were also used by Khan~\cite{Khan17} who presented a semi-streaming
algorithm that uses using $O(n)$ space for maintaining dynamic DFS of an undirected graph, 
requiring $O(\log^2 n)$ passes per update.
%

%


\subsection{Our Results}
\label{sec:OurResults}
We present the first semi-streaming algorithms to compute a DFS tree on an undirected graph 
in $o(n)$ passes. Our first result can be described using the following theorem.

\begin{theorem}
Given an undirected graph $G=(V,E)$, the DFS tree of the graph can be computed by a 
semi-streaming algorithm in at most $n/k$ passes using $O(nk)$ space, requiring $O(m\alpha(m,n))$ time per pass.
\label{thm:AdvAlg}
\end{theorem}

As described earlier, a simple algorithm can compute the DFS tree in $O(h)$ passes, where $h$ is the height of the DFS tree. Thus, for the graphs having a DFS tree with height $h=o(n)$ (see Appendix~\ref{apn:advAlg1WC} for details), we improve our result for such graphs in the following theorem.

\begin{theorem}
Given an undirected graph $G$, a DFS tree of $G$ can be computed by a 
semi-streaming algorithm using $\lceil h/k \rceil$ passes using $O(nk)$ space 
requiring amortized $O(m+nk)$ time per pass for any integer 
$k\leq h$, where $h$ is the height of the computed DFS tree.\footnote{
Note that there can be many DFS trees of a graph having varying heights, say $h_{min}$ to $h_{max}$. 
Our algorithm does not guarantee the computation of DFS tree having minimum height $h_{min}$, 
rather it simply computes a {\em valid} DFS tree with height $h$, where $h_{min}\leq h\leq h_{max}$.
}
\label{thm:FinalAlg}
\end{theorem}


Since typically the space allowed in the semi-streaming model is $O(n~poly\log n)$, 
the improvement in upper bounds of the problem by our results is considerably small 
(upto $poly\log n$ factors).
Recently, Elkin~\cite{Elkin17} presented the first $o(n)$ pass algorithm for computing 
Shortest Paths Trees. Using $O(nk)$ local space, it computes the shortest path tree
from a given source in $O(n/k)$ passes for unweighted graphs, and in $O(n\log n/k)$ passes
for weighted graphs. The significance of such results, despite improving the upper bounds 
by only small factors, is substantial because they address fundamental problems. 
The lack of any progress for such fundamental problems despite several decades of research on streaming 
algorithms further highlights the significance of such results.
Moreover, allowing $O(n^{1+\epsilon})$ space (as in \cite{FeigenbaumKMSZ05b,Baswana08,Elkin11})
such results improves the upper bound significantly by $O(n^{\epsilon})$ factors. 
Furthermore, they demonstrate the existence of a trade-off between the space and number of passes 
required for computing such fundamental structures.

Our final algorithm is presented in form of a {\em framework}, which can use any algorithm for 
maintaining a DFS tree of a stored sparser subgraph, 
provided that it satisfies the property of {\em monotonic fall}. Such a {\em framework}
allows more flexibility and is hopefully much easier to extend to better algorithms for computing 
a DFS tree or other problems requiring a computation of DFS tree. Hence we believe our {\em framework}
would be of independent interest.

We also augment our theoretical analysis with the experimental evaluation of our proposed algorithms. For both random and real graphs, the algorithms require merely a {\em few} passes even when the allowed space is just $O(n)$. The exceptional performance and surprising observations of our experiments on random graphs might also be of independent interest.

%

\subsection{Overview}
We now briefly describe the outline of our paper. In Section~\ref{sec:prelim} we establish the 
terminology and notations used in the remainder of the paper. In order to present the main ideas 
behind our approach in a simple and comprehensible manner, we present the algorithm in 
four stages. {\em Firstly} in Section~\ref{sec:alg1}, we describe the basic algorithm 
to build a DFS tree in $n$ passes, which adds a new vertex to the DFS tree in every pass over 
the input stream. {\em Secondly} in Section~\ref{sec:alg2}, we improve this algorithm to compute a DFS tree in $h$ passes, where $h$ is the height of the final DFS tree. This algorithm essentially 
computes all the vertices in the next level of the currently built DFS tree simultaneously, 
building the DFS tree by one level in each pass over the input stream. Thus, in the $i^{th}$ 
pass every vertex on the $i^{th}$ level of the DFS tree is computed.
{\em Thirdly} in Section~\ref{sec:alg3}, we describe an advanced algorithm which uses $O(nk)$ 
space to add a path of length at least $k$ to the DFS tree in every pass over the input stream. 
Thus, the complete DFS tree can be computed in $\lceil n/k \rceil$ passes.
{\em Finally}, in Section~\ref{sec:alg4}, we improve the algorithm to simultaneously add 
all the subtrees constituting the next $k$ levels of the final DFS tree starting from the leaves 
of the current tree $T$. Thus, $k$ levels are added to the DFS tree in each pass over the input stream, computing the DFS tree in $\lceil h/k \rceil$ passes. 
As described earlier, our final algorithm is presented in form of a {\em framework}
which uses as a black box, any algorithm to maintain a DFS tree of a stored sparser subgraph, 
satisfying certain properties. In the interest of completeness, one such algorithm is described in the Appendix~\ref{sec:rebuild}. {\em Lastly} in Section~\ref{sec:exp}, we present the results of the experimental evaluation of these algorithms. The details of this evaluation are deferred to Appendix~\ref{apn:expEval}.


In our advanced algorithms, we employ two interesting properties of a DFS tree, namely,
the {\em components} property~\cite{BaswanaCCK16} and the {\em min-height} property.
These simple properties of any DFS tree prove crucial in building the DFS efficiently
in the streaming environment. 

\section{Preliminaries}
\label{sec:prelim}
Let $G=(V,E)$ be an undirected connected graph having $n$ vertices and $m$ edges.
The DFS traversal of $G$ starting from any vertex $r\in V$ 
produces a spanning tree rooted at $r$ 
called a DFS tree, in $O(m+n)$ time. 
For any rooted spanning tree of $G$, a non-tree edge of the graph
is called a {\it back edge} if one of its endpoints is an ancestor of the 
other in the tree, else it is called a {\it cross edge}. 
A necessary and sufficient condition for any rooted spanning tree to be a 
DFS tree is that every non-tree edge is a back edge. 


In order to handle disconnected graphs, we add a dummy vertex $r$ to the graph
and connect it to all vertices. Our algorithm computes a DFS tree rooted at 
$r$ in this augmented graph, where each child subtree of $r$ is a DFS tree of a 
connected component in the DFS forest of the original graph. 
The following notations will be used throughout the paper.
\begin{itemize}
\itemsep0em 
\item $T:$~ The DFS tree of $G$ incrementally computed by our algorithm.
\item $par(w):$~ Parent of $w$ in $T$. 
\item  $T(x):$ The subtree of $T$ rooted at vertex $x$. 
\item $root(T'):$~ Root of a subtree $T'$ of $T$, i.e., $root\big(T(x)\big)=x$.
\item $level(v):$ Level of vertex $v$ in $T$, where $level(root(T))=0$ and
$level(v)=level(par(v))+1$.
\end{itemize}

In this paper we will discuss algorithms to compute a DFS tree $T$ for the input graph $G$
in the semi-streaming model. In all the cases $T$ will be built iteratively starting from an empty tree.
At any time during the algorithm, we shall refer to the vertices that are not a part of 
the DFS tree $T$ as {\em unvisited} and denote them by $V'$, i.e., $V'=V\setminus T$. 
Similarly, we refer to the subgraph induced by the {\em unvisited} vertices, $G'=G(V')$, 
as the {\em unvisited graph}.
Unless stated otherwise, we shall refer to a connected component of the unvisited graph $G'$
as simply a {\em component}. 
For any component $C$, the set of edges and vertices in the component will be denoted by $E_C$ and $V_C$.
Further, each component $C$ maintains a spanning tree of the component that shall be  referred as $T_C$. 
We refer to a path $p$ in a DFS tree $T$ as an {\em ancestor-descendant} path if one of 
its endpoints is an ancestor of the other in $T$.
Since the DFS tree grows downwards from the root, 
a vertex $u$ is said to be {\em higher} than vertex $v$ if $level(u)<level(v)$.
Similarly, among two edges incident on an ancestor-descendant path $p$, an edge $(x,y)$ is {\em higher}
than edge $(u,v)$ if $y,v\in p$ and $level(y)<level(v)$.

We shall now describe two invariants such that any algorithm computing DFS tree incrementally 
satisfying these invariants at every stage of the algorithm, 
ensures the absence of cross edges in $T$ and hence the correctness of the final DFS tree $T$. \\

{\centering
	\fbox{\parbox{\linewidth}{
			{\bf Invariants:}
			\begin{enumerate}[leftmargin=1cm]
				\item[${\cal I}_1:$] All non-tree edges among vertices in $T$ are back edges, and
				\item[${\cal I}_2:$] For any component $C$ of the unvisited graph, all the edges 
									from $C$ to the partially built DFS tree $T$ are incident on a single 
					  				{\em ancestor-descendant} path of $T$. 		
			\end{enumerate}
		}}
}
\vspace{.25em}

We shall also use the {\em components property} by Baswana et al.~\cite{BaswanaCCK16}, 
described as follows. 

\begin{figure}[!ht]
\centering
\setlength{\belowcaptionskip}{-10pt}
\includegraphics[width=0.3\linewidth]{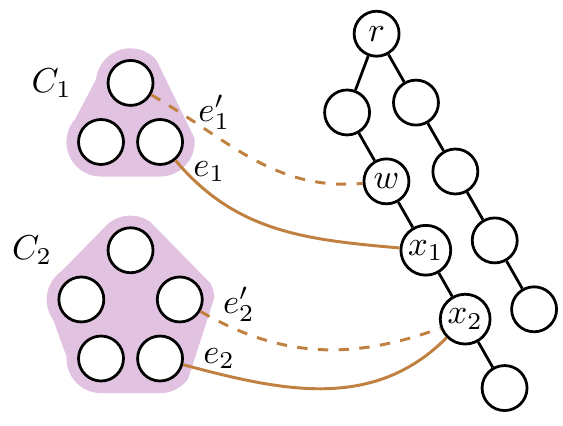}
\caption{Edges $e_1'$ and $e_2'$ can be ignored during the DFS traversal (reproduced from~\cite{BaswanaCCK16}).}
\label{fig:component-property}
\end{figure}

\begin{lemma}[Components Property~\cite{BaswanaCCK16}]
\label{lemma:L()}
Consider a partially completed DFS traversal where $T$ is the partially built DFS tree.
Let the connected components of $G'$ be $C_1,..,C_k$. 
Consider any two edges $e_i$ and $e'_i$ from $C_i$ that 
are incident respectively on a vertex $x_{i}$ and some 
ancestor (not necessarily proper) $w$ of $x_{i}$ in $T$. 
Then it is sufficient to consider only $e_i$ during the DFS traversal, 
i.e., the edge $e'_i$ can be safely ignored. 
%
\end{lemma}

Ignoring $e'_i$ during the DFS traversal, as stated in the components property, is justified 
because $e'_i$ will appear as a back edge in the resulting DFS tree 
(refer to Figure~\ref{fig:component-property}). For each component $C_i$, 
the edge $e_i$ can be found using a single pass over all the graph edges.

\section{Simple Algorithms}
\label{sec:alg1}
We shall first briefly describe the trivial algorithm to compute a DFS tree of a 
(directed) graph using $n$ passes. Since we are limited to have only $O(n~poly\log n)$ space, 
we cannot store the adjacency list of the
vertices in the graph. Recall that in the standard DFS algorithm~\cite{Tarjan72}, 
after visiting a vertex $v$, we choose any unvisited neighbour of $v$ and visit it.
If no neighbour of $v$ is unvisited, the traversal retreats back to the 
parent of $v$ and look for its unvisited neighbour, and so on. 

In the streaming model, we can use the same algorithm.
However, we do not store the adjacency list of a vertex.
To find the unvisited neighbour of each vertex, we perform a complete pass over the edges in $E$. 
The algorithm only stores the partially built DFS tree and the status of each vertex 
(whether it is visited/added to $T$). 
Thus, for each vertex $v$ (except $r$) one pass is performed to add $v$ to $T$ and 
another is performed before retreating to the parent of $v$.
Hence, it takes $2(n-1)$ passes to complete the algorithm 
since $T$ is initialized with the root $r$. 
Since, this procedure essentially simulates the standard DFS algorithm~\cite{Tarjan72}, 
it clearly satisfies the invariants ${\cal I}_1$ and ${\cal I}_2$.

This procedure can be easily transformed to require only $n-1$ passes by 
avoiding an extra pass for retreating from each vertex $v$.
In each pass we find an edge $e$ (from the stream) from the unvisited vertices, $V'$, to  
the lowest vertex on the ancestor-descendant path connecting $r$ and $v$, i.e., closest to $v$. 
Hence $e$ would be an edge from the lowest (maximum level) ancestor of $v$ (not necessarily proper) 
having at least one unvisited neighbour. 
Recall that if $v$ does not have an unvisited neighbour we move to processing its 
parent, and so on until we find an ancestor having an unvisited neighbour. 
We can thus directly add the edge $e$ to $T$.
Hence, retreating from a vertex would not require an additional pass and 
the overall procedure can be completed in 
$n-1$ passes, each pass adding a new vertex to $T$. 
Moreover, this also requires $O(1)$ processing time per edge and 
extra $O(n)$ time at the end of the pass, to find the relevant ancestor.
Refer to Procedure~\ref{alg:simple1} in Appendix~\ref{sec:pseudoCode} for the pseudocode of the procedure.
Thus, we get the following result.

\begin{theorem}
Given a directed/undirected graph $G$, a DFS tree of $G$ can be computed by a 
semi-streaming algorithm in $n$ passes using $O(n)$ space, using $O(m)$ time per pass.
\end{theorem}

\subsection{Improved algorithm}
\label{sec:alg2}
We shall now describe how this simple algorithm can be improved to compute a DFS tree 
of an undirected graph in $h$ passes, where $h$ is the height of the computed DFS tree. 
The main idea behind this approach is that each component of the unvisited graph $G'$ 
will constitute a separate subtree of the final DFS tree. 
Hence each such subtree can be computed independent of each other 
in parallel (this idea was also used by \cite{Khan17}).

Using one pass over edges in $E$, the components of the unvisited graph $G'$ 
can be found by using Union-Find algorithm~\cite{Tarjan75,TarjanL84} on the edges $E'$ of $G'$. 
Now, using the \textit{components} property we know that it is sufficient 
to add the lowest edge from each component to the DFS tree $T$.
At the end of the pass, for each component $C$ we find the edge $(x_C,y_C)$ incident from
the lowest vertex $x_C\in T$ to some vertex $y_C \in V_C$ and add it to $T$.
Note that in the next pass, for each component of $C\setminus\{y_C\}$ the lowest edge 
connecting it to $T$ would necessarily be incident on $y_C$ as $C$ was connected. 
Hence, instead of lowest edge incident on $T$, 
we store $e_y$ from $y\in V'$ only if $e_y$ is incident on some leaf of $T$.
Refer to Procedure~\ref{alg:simple} in Appendix~\ref{sec:pseudoCode} for the pseudocode of the algorithm.

To prove the correctness of the algorithm, we shall prove using induction that 
the invariants ${\cal I}_1$ and ${\cal I}_2$ hold over the passes performed on $E$. Since $T$ is initialized as an 
isolated vertex $r$, both invariants trivially hold. 
Now, let the invariants hold at the beginning of a pass.
Using $\II_2$, each component $C$ can have edges to a single ancestor-descendant path from $r$ to $x_C$. 
Thus, adding the edge $(x_C,y_C)$ for each component $C$, would not violate $\II_1$ at the
end of the pass,
given that $\II_1$ holds at the beginning of the pass.
Additionally, from each component $C$ we add a single vertex $y_C$ as a child of $x_C$ to $T$. 
Hence for any component of $C\setminus\{y_C\}$, the edges to $T$ can only be to ancestors of 
$y_C$ (using $\II_2$ of previous pass), and an edge necessarily to $y_C$,
satisfying $\II_2$ at the end of the pass.
Hence, using induction both $\II_1$ and $\II_2$ are satisfied 
proving the correctness of our algorithm.
 
Further, since each component $C$ in any $i^{th}$ pass necessarily has an edge to  
a leaf $x_C$ of $T$, the new vertex $y_C$ is added to the 
$i^{th}$ level of $T$. This also implies that every vertex at $i^{th}$ level of the 
final DFS tree is added during the $i^{th}$ pass. 
Hence, after $h$ passes we get a DFS tree of the whole graph as $h$ is the height
of the computed DFS tree. Now, the total time\footnote{
The Union-Find algorithm~\cite{Tarjan75,TarjanL84} requires $O(m\alpha(m,n))$ time, where $\alpha(m,n)$
is the inverse Ackermann function} 
required to compute the connected components is $O(m\alpha(m,n))$. And computing an edge from each unvisited vertex 
to a leaf in $T$ requires $O(1)$ time using $O(n)$ space. Thus, we have the following theorem.

\begin{theorem}
Given an undirected graph $G$, a DFS tree of $G$ can be computed by a 
semi-streaming algorithm in $h$ passes using $O(n)$ space, where $h$ is the height of the computed
DFS tree, using $O(m\alpha(m,n))$ time per pass.
\end{theorem}

\section{Computing DFS in sublinear number of passes}
\label{sec:alg3}
Since a DFS tree may have $O(n)$ height, we cannot hope to compute
a DFS tree in sublinear number of passes using the previously described simple algorithms.  
The main difference between the advanced approaches and the simple algorithms
is that, in each pass instead of adding a single vertex (say $y$) to the DFS tree, 
we shall be adding an entire path (starting from $y$) to the DFS tree. 
The DFS traversal gives the flexibility to chose the next vertex
to be visited as long as the DFS property is satisfied, i.e., invariants $\II_1$ and $\II_2$
are maintained.

Hence in each pass we do the following for every component $C$ in $G'$. 
Instead of finding a single edge $(x_C,y_C)$ (see Section~\ref{sec:alg2}),
we find a path $P$ starting from $y_C$ in $C$ and attach this entire path $P$ to $T$ 
(instead of only $y_C$). Suppose this splits the component $C$ into components 
$C_1,C_2,\dots$ of $C\setminus P$. Now, each $C_i$ would have an edge to at least 
one vertex on $P$ (instead of necessarily the leaf $x_C$ in Section~\ref{sec:alg2}) 
since $C$ was a connected component. Hence in this algorithm for each $C_i$,
we find the vertex $y_i$ which is the lowest vertex of $T$ (or $P$) to which an edge from $C_i$
is incident. Observe that $y_i$ is unique since all the neighbours of $C_i$ in $T$ 
are along one path from the root to a leaf. 
Using the {\em components} property, the selection of $y_i$ as the parent of the root of the subtree to be computed for $C_i$ ensures that invariant ${\cal I}_2$ continues to hold.
Thus, in each pass from every component of the unvisited graph, 
we shall extract a path and add it to the DFS tree $T$.

This approach thus allows $T$ to grow by more than one vertex in each pass which is essential
for completing the tree in $o(n)$ passes. If in each pass we add a path of length at least $k$ 
from each component of $G'$, then the tree will grow by at least $k$ vertices in each pass, 
requiring overall $\lceil n/k \rceil$ passes to completely build the DFS tree.
We shall now present an important  property of any DFS tree of an undirected graph, 
which ensures that in each pass we can find a path of length at least $k\geq m/n$ (refer to Appendix~\ref{sec:proofs} for proof).


\begin{lemma}[Min-Height Property]
Given a connected undirected graph $G$ having $m$ edges, any DFS tree of $G$ 
from any root vertex necessarily has a height $h\geq m/n$.
\label{lem:dfsHeight} 
\end{lemma}

\subsection{Algorithm}
We shall now describe our algorithm to compute a DFS tree of the input graph in $o(n)$ passes.
Let the maximum space allowed for computation in the semi-streaming model be $O(nk)$.
The algorithm is a recursive procedure that computes a DFS tree of a component $C$
from a root $r_C$. For each component $C$ we maintain a spanning tree $T_C$ of $C$. 
Initially we can perform a single pass over $E$ to compute a spanning tree of the connected
graph $G$ (recall the assumption in Section~\ref{sec:prelim}) using the Union-Find algorithm. 
For the remaining components, its spanning tree would already have been computed and passed
as an input to the algorithm.


We initiate a pass over the edge in $E$ and store the first $|V_C|\cdot k$ edges (if possible) 
from the component $C$ in the memory as the subgraph $E'_C$. 
Before proceeding with the remaining stream, we use any algorithm for computing a DFS tree 
$T'_C$ rooted at $r_C$ in the subgraph containing edges from $T_C$ and $E'_C$. 
Note that adding $T_C$ to $E'_C$ is important to ensure that subgraph induced by 
$T_C\cup E'_C$ is connected. In case the pass was completed before $E'_C$ exceeded
storing $|V_C|\cdot k$ edges, $T'_C$ is indeed a DFS tree of $C$ and we directly add it to $T$.
Otherwise, we find the longest path $P$ from $T'_C$ starting from $r_C$, i.e., path from $r_C$ 
to the farthest leaf. The path $P$ is then added to $T$.

Now, we need to compute the connected components of $C\setminus P$ and the new corresponding root
for each such component. We use the Union-Find algorithm to compute these components, 
say $C_1,...,C_f$, and compute the lowest edge $e_i$ from each $C_i$ on the path $P$. 
Clearly, there exist such an edge as $C$ was connected. In order to find these components 
and edges, we need to consider all the edges in $E_C$, which can be done by first considering 
$E'_C$ and then each edge from $C$ in the remainder of input stream of the pass. 
Refer to Procedure~\ref{alg:advAlg1} in Appendix~\ref{sec:pseudoCode} for the pseudocode of the algorithm.

%
%
%
%
%

Using the components property, choosing the new root $y_i$ corresponding to the lowest edge 
$e_i$ ensures that the invariant $\II_2$ and hence $\II_1$ is satisfied.
Now, in case $|E_C|< |V_C|\cdot k$, the entire DFS tree of $C$ is constructed and added to $T$ in a single pass.
Otherwise, in each pass we add the longest path $P$ from $T'_C$ to the final DFS tree $T$. 
Since $|E'_C|=|V_C|\cdot k$ and $E'_C\cup T_C$ is a single connected component,
the {\em min-height property} ensures that the height of any such $T'_C$ (and hence $P$) is at least $k$.
Since in each pass, except the last, we add at least $k$ new vertices to $T$, 
this algorithm terminates in at most $\lceil n/k \rceil$ passes.
Now, the total time required to find the components of the unvisited graph is again $O(m\alpha(m,n))$.
The remaining operations clearly require $O(|E_C|)$ time for a component $C$, requiring overall $O(m)$ time.
Thus, we get the following theorem. 


\newtheorem*{thmAdvAlg}{Theorem~\ref{thm:AdvAlg}}
\begin{thmAdvAlg}
Given an undirected graph $G$, a DFS tree of $G$ can be computed by a 
semi-streaming algorithm in at most $\lceil n/k \rceil$ passes using $O(nk)$ space,
requiring $O(m\alpha(m,n))$ time per pass.
\end{thmAdvAlg}

\begin{remark}
Since, Procedure~\ref{alg:advAlg1} adds an ancestor-descendant path for each component of $G'$, 
it might seem that the analysis of the algorithm is not tight for computing DFS trees with $o(n)$ height.
However, there exist a sequence of input edges where Procedure~\ref{alg:advAlg1} indeed takes $\Theta(n/k)$ passes
for computing a DFS tree with height $o(n)$ (see Appendix~\ref{apn:advAlg1WC}). 
\end{remark}

\section{Final algorithm}
\label{sec:alg4}

We shall now further improve the algorithm so that the required number of passes reduces 
to $\lceil h/k \rceil$  while it continues to use $O(nk)$ space, where $h$ is the height 
of the computed DFS tree and $k$ is any positive integer. To understand the main intuition 
behind our approach, let us recall the previously described algorithms. 
We first described a simple algorithm (in Section~\ref{sec:alg1}) in which every pass over 
the input stream adds one new vertex as the child of some leaf of $T$, which was improved 
(in Section~\ref{sec:alg2}) to simultaneously adding all vertices which are children of the 
leaves of $T$ in the final DFS tree. 
We then presented another algorithm (in Section \ref{sec:alg3}) in which every pass over the 
input stream adds one {\em ancestor-descendant} path of length $k$ or more, from each component 
of $G'$ to $T$. We shall now improve it by adding all the subtrees constituting the next $k$ 
levels of the final DFS tree starting from the leaves of the current tree $T$
(or fewer than $k$ levels if the corresponding component of $G'$ is exhausted).

Now, consider any component $C$ of $G'$. Let $r_C\in C$ be a vertex having an edge $e$ to a 
leaf of the partially built DFS tree $T$. The computation of $T$ can be completed by computing 
a DFS tree of $C$ from the root $r_C$, which can be directly attached to $T$ using $e$.
However, computing the entire DFS tree of $C$ may not be possible in a single pass over the 
input stream, due to the limited storage space available. Thus, using $O(n\cdot k)$ space we 
compute a special spanning tree $T_C$ for each component $C$ of $G'$ in parallel, 
such that the top $k$ levels of $T_C$ is same as the top $k$  levels of {\em some} DFS tree of $C$. 
As a result, in the $i^{th}$ pass all vertices on the levels $(i-1)\cdot k+1$ to $i\cdot k$ of 
the final DFS tree are added to $T$. This essentially adds a tree $T'_C$ representing the top $k$ 
levels of $T_C$ for each component $C$ of $G'$.  This ensures that our algorithm will terminate 
in $\lceil h/k \rceil$ passes, where $h$ is the height of the final DFS tree. Further, this special 
tree $T_C$ also ensures an \textit{additional} property, i.e., there is a one to one correspondence between 
the set of trees of $T_C\setminus T'_C$ and the components of $C\setminus T'_C$. 
In fact, each tree of $T_C\setminus T'_C$ is a spanning tree of the corresponding component. 
This property directly provides the spanning trees of the components of $G'$ in the next pass.

\subsubsection*{Special spanning tree $T_C$}
We shall now describe the properties of this special tree $T_C$ (and hence $T'_C$) which 
is computed in a single pass over the input stream. For $T'_C$ to be added to the DFS tree 
$T$ of the graph, a necessary and sufficient condition is that $T'_C$ satisfies the 
invariants ${\cal I}_1$ and ${\cal I}_2$ at the end of the pass. 
To achieve this we maintain $T_C$ to be a spanning tree of $C$, such that these invariants are 
maintained by the corresponding $T'_C$ throughout the pass as the edges are processed. 
Let $S_C$ be the set of edges already visited during the current pass, which have both endpoints in $C$.
In order to satisfy ${\cal I}_1$, no edge in $S_C$ should be a cross edge in $T'_C$, i.e., 
no edge having both endpoints in the top $k$ levels of $T_C$ is a cross edge. 
In order to satisfy ${\cal I}_2$, no edge in $S_C$ from any component $C'\in C\setminus T'_C$
to $C\setminus C'$ should be a cross edge in $T_C$.  Hence, using the {\em additional} property of $T_C$, 
each edge from a tree $\tau$ in $T_C\setminus T'_C$ to $T_C\setminus \tau$ is necessarily a back edge. 
This is captured by the two conditions of invariant ${\cal I}_T$ given below. 
Hence ${\cal I}_T$ should hold after processing each edge in the pass. 
Observe that any spanning tree, $T_C$, trivially satisfies ${\cal I}_T$ at  the beginning of the pass 
as $S_C=\emptyset$.\\

{\centering
        \fbox{\parbox{\linewidth}{
{\bf Invariant ${\cal I}_T:$}\\ 
$T_C$ is a spanning tree of $C$ with the top $k$ levels being $T'_C$ such that:
\begin{enumerate}[leftmargin=1cm]
\item[${\cal I}_{T_1}:$] All non-tree edges of $S_C$ having both endpoints in $T'_C$, are back edges.
\item[${\cal I}_{T_2}:$] For each tree $\tau$ in $T_C\setminus T'_C$, all the edges of $S_C$ from 
		$\tau$ to $T_C\setminus \tau$ are back edges. 
\end{enumerate} 
}}}
\vspace{.25em}

Thus, ${\cal I}_T$ is the local invariant maintained by $T_C$ during the pass, 
so that the global invariants $\II_1$ and $\II_2$ are maintained throughout the algorithm.
Now, in order to compute $T_C$ (and hence $T'_C$) satisfying the above invariant, we store 
a subset of  $S_C$ along with $T_C$. Let $H_C$ denote the (spanning) subgraph of $G$ formed 
by $T_C$ along with these additional edges. Note that all the 
edges of $S_C$ cannot be stored in $H_C$ due to space limitation of $O(nk)$.
Since each pass starts with the spanning tree $T_C$ of $C$ and $S_C=\emptyset$, initially $H_C=T_C$.
As the successive edges of the stream are processed, $H_C$ is updated if the input edge belongs
to the component $C$. We now formally describe $H_C$ and its properties.

\subsubsection*{Spanning subgraph $H_C$}
As described earlier, at the beginning of a pass for every component $C$ of $G'$,
$H_C = T_C$. Now, the role of $H_C$ is to facilitate
the maintenance of the invariant ${\cal I}_T$. In order to satisfy ${\cal I}_{T_1}$
and ${\cal I}_{T_2}$, we store in $H_C$ all the edges in $S_C$ that are incident on at least one vertex of $T'_C$. 
Therefore, $H_C$ is the spanning tree $T_C$ 
along with every edge in $S_C$ which has at least one endpoint in $T'_C$.
Thus, $H_C$ satisfies the following invariant throughout the algorithm.\\

{\centering
        \fbox{\parbox{\linewidth}{
{\bf Invariant ${\cal I}_H:$}\\
$H_C$ comprises of $T_C$ 
and all edges from $S_C$ that are incident on at least one vertex of $T'_C$.  }}}
\vspace{.25em}

We shall now describe a few properties of $H_C$ and then in the following section
show that maintaining ${\cal I}_H$ for $H_C$ is indeed sufficient to maintain the invariant ${\cal I}_T$
as the stream is processed. 
The following properties of $H_C$ are crucial to establish the correctness of our procedure to maintain 
$T_C$ and $H_C$ and establish a bound on total space required by $H_C$ (see Appendix~\ref{sec:proofs} for proofs).

\begin{lemma}
$T_C$ is a valid DFS tree of $H_C$.
\label{lem:dfsProp}
\end{lemma}

\begin{lemma}
The total number of edges in $H_C$, for all the components $C$ of $G'$, is $O(nk)$. 
\label{lem:HCspace}
\end{lemma}

\subsection{Processing of Edges}
We now describe how $T_C$ and $H_C$ are maintained while processing the edges of the 
input stream such that ${\cal I}_T$ and ${\cal I}_H$ are satisfied. 
Since our algorithm maintains the invariants 
${\cal I}_1$ and ${\cal I}_2$ (because of ${\cal I}_T$), we know that any edge whose both endpoints are not in some 
component $C$ of $G'$, is either a back edge or already a tree edge in $T$. Thus, we shall only discuss 
the processing of an edge $(x,y)$ having both endpoints in $C$ (now added to $S_C$), 
where $level(x)\leq level(y)$.


\begin{enumerate}
\item If $x\in T'_C$ then the edge is added to $H_C$ to ensure ${\cal I}_H$.
In addition, if $(x,y)$ is a cross edge in $T_C$ it violates either ${\cal I}_{T_1}$ (if $y\in T'_C$)
or ${\cal I}_{T_2}$ (if $y\notin T'_C$). 
Thus, $T_C$ is required to be restructured to ensure that ${\cal I}_T$ is satisfied. 

\item If $x\notin T'_C$ and if $x$ and $y$ belong to different trees in $T_C\setminus T'_C$, 
then it violates ${\cal I}_{T_2}$.
Again in such a case, $T_C$ is required to be restructured to ensure that ${\cal I}_T$ is satisfied. 
\end{enumerate}

Note that after restructuring $T_C$ we need to update $H_C$ such that ${\cal I}_H$ is satisfied.
Consequently any non-tree edge in $H_C$ that was incident on a vertex in original $T'_C$,
has to be removed from $H_C$ if none of its endpoints are in $T'_C$ after restructuring $T_C$,
i.e., one or both of its endpoints have moved out of $T'_C$. 
But the problem arises 
if a vertex moves into $T'_C$ during restructuring.
There might have been edges incident on such a vertex in $S_C$ and which were not stored in $H_C$.
In this case we need these edges in $H_C$ to satisfy ${\cal I}_H$, which is not possible
without visiting $S_C$ again. 
This problem can be avoided if our restructuring procedure ensures that no new vertex enters $T'_C$.
This can be ensured if the restructuring procedure follows the property of {\em monotonic fall}, i.e.,
the level of a vertex is never decreased by the procedure. 
Let $e$ be the new edge of component $C$ in the input stream. 
We shall show that in order to preserve the invariants ${\cal I}_T$ and ${\cal I}_H$
it is sufficient that the restructuring procedure maintains the property of {\em monotonic fall} and
ensures that the restructured $T_C$ is a DFS tree of $H_C+e$.

\begin{lemma}
On insertion of an edge $e$, any restructuring procedure which updates $T_C$ to be a valid DFS tree 
of $H_C + e$ ensuring {\em monotonic fall}, satisfies the invariants ${\cal I}_T$ and ${\cal I}_H$.
\label{lem:correctness}
\end{lemma}

\begin{proof}
The property of {\em monotonic fall} ensures that the vertex set of new $T'_C$ is a subset 
of the vertex set of the previous $T'_C$. Using ${\cal I}_H$ we know that any edge of $S_C$ 
which is not present in $H_C$ must have both its endpoints outside $T'_C$. Hence, 
{\em monotonic fall} guarantees that ${\cal I}_H$ continues to hold with respect to the new 
$T'_C$ for the edges in $S_C\setminus\{e\}$. Additionally, we save $e$ in the new $H_C$ 
if at least one of its endpoints belong to the new $T'_C$, ensuring that ${\cal I}_H$ 
holds for the entire $S_C$. 

Since the restructuring procedure ensures that the updated $T_C$ is a DFS tree of $H_C$, 
the invariant ${\cal I}_{T_1}$ trivially holds as a result of $\II_H$. In order to prove 
the invariant ${\cal I}_{T_2}$, consider any edge $e'\in S_C$ from a tree 
$\tau\in T_C\setminus T'_C$ to $T_C\setminus \tau$. Clearly, it will satisfy ${\cal I}_{T_2}$ if 
$e'\in H_C$, as $T_C$ is a DFS tree of $H_C+e$. In case $e' \notin H_C$, it must be internal to
some tree $\tau'$ in the original $T_C\setminus T'_C$ (using $\II_{T_2}$ in the original $T_C$). 
We shall now show that such an edge will remain internal to some tree in the updated 
$T_C\setminus T'_C$ as well, thereby not violating ${\cal I}_{T_2}$. Clearly the endpoints of 
$e'$ cannot be in the updated $T'_C$ due to the property of {\em monotonic fall}. 

Assume that the endpoints of $e'$ belong to different trees of updated $T_C \setminus T'_C$.
Now, consider the edges $e_1,...,e_t$ on the tree path in $\tau'$ connecting the endpoints of $e'$.
Since the entire tree path is in $\tau'$, the endpoints of each $e_i$ are not in original $T'_C$, 
ensuring that they are also not in the updated $T'_C$ (by {\em monotonic fall}). Since the endpoints
of $e'$ (and hence the endpoints of the path $e_1,...,e_t$) are in different trees in updated 
$T_C \setminus T'_C$, there must exist some $e_i$ which also has endpoints belonging to different 
trees of updated $T_C\setminus T'_C$. 
This makes $e_i$ a cross edge of the updated $T_C$. Since $e_i$ is a tree edge 
of original $T_C$, it belongs to $H_C$ and hence $e_i$ being a cross edge implies that 
the updated $T_C$ is not a DFS tree of $H_C+e$, which is a contradiction. 
Hence $e'$ has both its endpoints in the same tree of the updated $T_c\setminus T'_C$,
ensuring that ${\cal I}_{T_2}$ holds after the restructuring procedure. 
\end{proof}

Hence, any procedure to restructure a DFS tree $T_C$ of the subgraph $H_C$ on insertion of a 
{\em cross edge} $e$, that upholds the property of {\em monotonic fall} and returns a new $T_C$ 
which is a DFS tree of $H_C+e$, can be used as a {\em black box} in our algorithm. One such algorithm is the incremental 
DFS algorithm by Baswana and Khan~\cite{BaswanaK17}, which precisely fulfils our requirement. 
They proved the total update time of the algorithm to be $O(n^2)$. They also showed that any 
algorithm maintaining incremental DFS abiding {\em monotonic fall} would require $\Omega(n^2)$ time 
even for sparse graphs, if it explicitly maintains the DFS tree. If the height of the DFS tree is 
known to be $h$, these bounds reduces to $O(nh+n_e)$ and $\Omega(nh+n_e)$ respectively, 
where $n_e$ is the number of edges processed by the algorithm. Refer to Appendix~\ref{sec:rebuild} 
for a brief description of the algorithm. 


\subsection{Algorithm}
We now describe the details of our final algorithm which uses Procedure~\ref{alg:rebuild}~\cite{BaswanaK17} (described in Appendix~\ref{sec:rebuild}) for restructuring 
the DFS tree when a cross edge is inserted. 
Similar to the algorithm in Section~\ref{sec:alg3}, for each component $C$ of $G'$, a rooted spanning tree $T_C$ of the component is required as an input to the procedure having the root $r_C$. 

Initially $T=\emptyset$ and  $G'=G$ has a single component $C$, as $G$ is connected 
(recall the assumption in Section~\ref{sec:prelim}). Hence for the first pass,
we compute a spanning tree $T_C$ of $G$ using the Union-Find algorithm. Subsequently 
in each pass we directly get a spanning tree $T_{C'}$ for each component $C'$ of the new $G'$, 
which is the corresponding tree in $T_C\setminus T'_C$, 
where $C$ is the component containing $C'$ in the previous pass.
Also, observe that the use of these trees as the new $T_C$ ensures that the level of no vertex ever rises
in the context of the entire tree $T$. This implies that the level of any vertex starting 
with the initial spanning tree $T_G$ never rises, i.e., the entire algorithm satisfies the property 
of {\em monotonic fall}. We will use this fact crucially in the analysis of the time complexity.

As described earlier, we process the edges of the stream by updating the $T_C$ and $H_C$ maintaining
$\II_T$ and $\II_H$ respectively. In case the edge is internal to some tree in $T_C\setminus T'_C$ 
(i.e., have both endpoint in the same tree in $T_C\setminus T'_C$), we simply ignore the edge.
Otherwise, we add it to $H_C$ to satisfy $\II_H$. Further, Procedure~\ref{alg:rebuild} 
maintains $T_C$ to be a DFS tree of $H_C$, which restructures $T_C$ if the processed edge is added 
to $H_C$ and is a cross edge in $T_C$. 
Now, in case $T_C$ is updated we also update the subgraph $H_C$, by removing the extra 
non-tree edges having both endpoints in $T_C\setminus T'_C$. After the pass is completed, 
we attach $T'_C$ (the top $k$ levels of $T_C$) to $T$. Now, $\II_{T_2}$ ensures that each tree 
of $T_C\setminus T'_C$ forms the (rooted) spanning tree of the components of the new $G'$, 
and hence can be used for the next pass. Refer to Procedure~\ref{alg:advAlg2} in Appendix~\ref{sec:pseudoCode}
for the pseudocode of the algorithm.

\subsection{Correctness and Analysis}

The correctness of our algorithm follows from Lemma~\ref{lem:correctness}, which ensures 
that invariants ${\cal I}_{H}$ and ${\cal I}_{T}$ (and hence ${\cal I}_1$ and ${\cal I}_2$) 
are maintained as a result of using Procedure~\ref{alg:rebuild} which ensures {\em monotonic fall} of vertices. 
The total space used by our algorithm and the restructuring procedure is dominated 
by the cumulative size of $H_C$ for all components $C$ of $G'$, which is $O(nk)$ using Lemma~\ref{lem:HCspace}.
Now, in every pass of the algorithm, a DFS tree for each component $C$ of height $k$ is 
attached to $T$. These trees collectively constitute the next $k$ levels of the final DFS tree $T$. 
Therefore, the entire tree $T$ is computed in $\lceil h/k \rceil$ passes.

Let us now analyse the time complexity of our algorithm. In the first pass $O(m\alpha(m,n))$ 
time is required to compute the spanning tree $T_C$ using the Union-Find algorithm. 
Also, in each pass $O(m)$ time is required to process the input stream. Further, 
in order to update $H_C$ we are required to delete edges 
having both endpoints out of $T'_C$. Hence, whenever a vertex falls below the $k^{th}$ level, the edges 
incident on it are checked for deletion from $H_C$ (if the other endpoint is also not in $T'_C$).
Total time required for this is $O(\sum_{v\in V} deg(v)) = O(m)$ per pass. 
Now, Appendix~\ref{sec:rebuild} describes the details of Procedure~\ref{alg:rebuild} which 
maintains the DFS tree in total $O(nh+n_e)$ time, where $n_e=O(mh/k)$, for processing the entire
input stream in each pass. 

Finally, we need to efficiently answer the {\em query} whether an edge is internal to some tree in 
$T_C\setminus T'_C$. For this we maintain for each vertex $x$ its ancestor at level $k$ 
as $rep[x]$, i.e., $rep[x]$ is the root of the tree  in $T_C\setminus T'_C$ 
that contains $x$. If $level(x) < k$, then $rep[x] =x$. For an edge $(x,y)$ comparing the 
$rep[x]$ and $rep[y]$ efficiently answers the required query in $O(1)$ time. 
However, whenever $T_C$ is updated we need to update 
$rep[v]$ for each vertex $v$ in the modified part of $T_C$, requiring $O(1)$ time per vertex 
in the modified part of $T_C$. We shall bound the total work done to update $rep[x]$ of 
such a vertex $x$ throughout the algorithm to $O(nh)$ as follows.

Consider the potential function $\Phi=\sum_{v\in V} level(v)$. Whenever some part of $T_C$ is updated, 
each vertex $x$ in the modified $T_C$ necessarily incurs a fall in its level 
(due to {\em monotonic fall}). 
Thus, the cost of updating $rep[x]$ throughout the algorithm is proportional to the number of times
$x$ descends in the tree, hence increases the value of $\Phi$ by at least one unit.
Hence, updating $rep[x]$ for all $x$ in the modified part of $T_C$ can be accounted by the 
corresponding increase in the value of $\Phi$. Clearly, the maximum value of $\Phi$ is $O(nh)$, 
since the level of each vertex is always less than $h$, where $h$ is the height of the 
computed DFS tree. Thus, the total work done to update $rep[x]$ for all $x\in V$ is $O(nh)$.
This proves our main theorem described in Section~\ref{sec:OurResults} which is stated as follows.

\newtheorem*{thmMain}{Theorem~\ref{thm:FinalAlg}}
\begin{thmMain}
Given an undirected graph $G$, a DFS tree of $G$ can be computed by a 
semi-streaming algorithm using $\lceil h/k \rceil$ passes using $O(nk)$ space 
requiring amortized $O(m+nk)$ time per pass for any integer 
$k\leq h$, where $h$ is the height of the computed DFS tree.
\end{thmMain}

\noindent
\textbf{Remark: } Note that the time complexity of our algorithm is indeed tight for our framework.
Since our algorithm requires $\lceil h/k \rceil $ passes and any restructuring procedure following 
{\em monotonic fall} requires $\Omega(nh+n_e)$ time, each pass would require $\Omega(m+nk)$ time.

\section{Experimental Evaluation}
\label{sec:exp}
Most streaming algorithms deal with only $O(n)$ space, for which our advanced algorithms improve over the simple algorithms theoretically by just constant factors. However, their empirical performance demonstrates their significance in the real world applications. The evaluation of our algorithms on random and real graphs shows that in practice these algorithms require merely a {\em few} passes even when allowed to store just $5n$ edges. The results of our analysis can be summarized as follows (for details refer to Appendix~\ref{apn:expEval}).


The two advanced algorithms {\em kPath} (Algorithm~\ref{alg:advAlg1} in Section~\ref{sec:alg3}) and {\em kLev} (Algorithm~\ref{alg:advAlg2} in Section~\ref{sec:alg4} with an additional heuristic) perform much better than the rest even when $O(n)$ space is allowed. For both random and real graphs, {\em kPath} performs slightly worse as the density of the graph increases. On the other hand {\em kLev} performs slightly better only in random graphs with the increasing density. The effect of the space parameter is very large on {\em kPath} from $k=1$ to small constants, requiring very {\em few} passes even for $k=5$ and $k=10$. However, {\em kLev} seems to work very well even for $k=1$ and has a negligible effect of increasing the value of $k$. Overall, the results suggest using {\em kPath} if $nk$ space is allowed for $k$ being a small constant such as $5$ or $10$. However, if the space restrictions are extremely tight it is better to use {\em kLev}.

%

\section{Conclusion}
We presented the first $o(n)$ pass semi-streaming algorithm for computing a DFS tree 
for an undirected graph, breaking the long standing presumed barrier of $n$ passes.
In our streaming model we assume that $O(nk)$ local space is available for computation,
where $k$ is any natural number.
Our algorithm computes a DFS tree in $\lceil n/k \rceil$ passes.
We improve our algorithm to require only $\lceil h/k \rceil$ passes without any additional space
requirement, where $h$ is the height of the final tree. This improvement becomes significant 
for graphs having shallow DFS trees. Moreover, our algorithm is described as a {\em framework} 
using a restructuring algorithm as a black box. This allows more flexibility 
to extend our algorithm for solving other problems requiring a computation of DFS tree in the streaming
environment. 

Recently, in a major breakthrough Elkin~\cite{Elkin17} presented the first $o(n)$ pass algorithm 
for computing Shortest Paths Tree from a single source. 
Using $O(nk)$ local space, it computes the shortest path tree
from a given source in $O(n/k)$ passes for unweighted graphs and in $O(n\log n/k)$ passes
for weighted graphs.




Despite the fact that these breakthroughs provide only minor improvements 
(typically $poly\log n$ factors), they are significant steps to pave a path in better 
understanding of such fundamental problems in the streaming environment. 
These simple improvements come after decades of the emergence of 
streaming algorithms for graph problems, where such problems were considered implicitly hard in the 
semi-streaming environment.
We thus believe that our result is a significant improvement over the known algorithm for 
computing a DFS tree in the streaming environment, and it can be a useful step in more involved 
algorithms that require the computation of a DFS tree.
%

Moreover, the experimental evaluation of our algorithms revealed exceptional performance of the advanced algorithms {\em kPath} and {\em kLev} (greatly affected by the {\em additional heuristic}). Thus, it would be interesting to further study these algorithms theoretically which seem to work extremely well in practice.


\bibliographystyle{plain}
\bibliography{paper}

\begin{thebibliography}{10}

\bibitem{PHamster16}
Hamsterster full network dataset -- {KONECT}, September 2016.

\bibitem{AhnG13}
Kook~Jin Ahn and Sudipto Guha.
\newblock Linear programming in the semi-streaming model with application to
  the maximum matching problem.
\newblock {\em Inf. Comput.}, 222:59--79, 2013.

\bibitem{AlonMS99}
Noga Alon, Yossi Matias, and Mario Szegedy.
\newblock The space complexity of approximating the frequency moments.
\newblock {\em Journal of Computer and System Sciences}, 58(1):137 -- 147,
  1999.

\bibitem{AusielloFL09}
Giorgio Ausiello, Donatella Firmani, and Luigi Laura.
\newblock Datastream computation of graph biconnectivity: Articulation points,
  bridges, and biconnected components.
\newblock In {\em Theoretical Computer Science, 11th Italian Conference,
  {ICTCS} 2009, Cremona, Italy, September 28-30, 2009, Proceedings.}, pages
  26--29, 2009.

\bibitem{AusielloFL12}
Giorgio Ausiello, Donatella Firmani, and Luigi Laura.
\newblock Real-time monitoring of undirected networks: Articulation points,
  bridges, and connected and biconnected components.
\newblock {\em Networks}, 59(3):275--288, 2012.

\bibitem{Bar-YossefKS02}
Ziv Bar-Yossef, Ravi Kumar, and D.~Sivakumar.
\newblock Reductions in streaming algorithms, with an application to counting
  triangles in graphs.
\newblock In {\em Proceedings of the Thirteenth Annual ACM-SIAM Symposium on
  Discrete Algorithms}, SODA '02, pages 623--632, 2002.

\bibitem{Baswana08}
Surender Baswana.
\newblock Streaming algorithm for graph spanners - single pass and constant
  processing time per edge.
\newblock {\em Inf. Process. Lett.}, 106(3):110--114, 2008.

\bibitem{BaswanaCCK16}
Surender Baswana, Shreejit~Ray Chaudhury, Keerti Choudhary, and Shahbaz Khan.
\newblock Dynamic {DFS} in undirected graphs: breaking the {O}(\emph{m})
  barrier.
\newblock In {\em {ACM-SIAM} Symposium on Discrete Algorithms, {SODA}}, pages
  730--739, 2016.

\bibitem{BaswanaGK18}
Surender Baswana, Ayush Goel, and Shahbaz Khan.
\newblock Incremental {DFS} algorithms: a theoretical and experimental study.
\newblock In {\em Proceedings of the Twenty-Ninth Annual {ACM-SIAM} Symposium
  on Discrete Algorithms, {SODA} 2018, New Orleans, LA, USA, January 7-10,
  2018}, pages 53--72, 2018.

\bibitem{BaswanaK17}
Surender Baswana and Shahbaz Khan.
\newblock Incremental algorithm for maintaining a {DFS} tree for undirected
  graphs.
\newblock {\em Algorithmica}, 79(2):466--483, 2017.

\bibitem{ThijsLMPH05}
Thijs Beuming, Lucy Skrabanek, Masha~Y. Niv, Piali Mukherjee, and Harel
  Weinstein.
\newblock {PDZBase}: A protein--protein interaction database for {PDZ}-domains.
\newblock {\em Bioinformatics}, 21(6):827--828, 2005.

\bibitem{MRAA04}
Mari\'an Bogu\~n\'a, Romualdo Pastor-Satorras, Albert D\'{\i}az-Guilera, and
  Alex Arenas.
\newblock Models of social networks based on social distance attachment.
\newblock {\em Phys. Rev. E}, 70:056122, Nov 2004.

\bibitem{Bollobas84}
B{\'{e}}la Bollob{\'{a}}s.
\newblock The evolution of random graphs.
\newblock {\em Transactions of the American Mathematical Society}, 286
  (1):257--274, 1984.

\bibitem{BorradaileMM14}
Glencora Borradaile, Claire Mathieu, and Theresa Migler.
\newblock Lower bounds for testing digraph connectivity with one-pass streaming
  algorithms.
\newblock {\em CoRR}, abs/1404.1323, 2014.

\bibitem{Chicago2}
D.~E. Boyce, K.~S. Chon, M.~E. Ferris, Y.~J. Lee, K-T. Lin, and R.~W. Eash.
\newblock Implementation and evaluation of combined models of urban travel and
  location on a sketch planning network.
\newblock {\em Chicago Area Transportation Study}, pages xii + 169, 1985.

\bibitem{BuchsbaumGW03}
Adam~L. Buchsbaum, Raffaele Giancarlo, and Jeffery Westbrook.
\newblock On finding common neighborhoods in massive graphs.
\newblock {\em Theor. Comput. Sci.}, 299(1-3):707--718, 2003.

\bibitem{ChoML11}
Eunjoon Cho, Seth~A. Myers, and Jure Leskovec.
\newblock Friendship and mobility: User movement in location-based social
  networks.
\newblock In {\em Proc. Int. Conf. on Knowledge Discovery and Data Mining},
  pages 1082--1090, 2011.

\bibitem{Chicago1}
R.~W. Eash, K.~S. Chon, Y.~J. Lee, and D.~E. Boyce.
\newblock Equilibrium traffic assignment on an aggregated highway network for
  sketch planning.
\newblock {\em Transportation Research Record}, 994:30--37, 1983.

\bibitem{Elkin11}
Michael Elkin.
\newblock Streaming and fully dynamic centralized algorithms for constructing
  and maintaining sparse spanners.
\newblock {\em {ACM} Trans. Algorithms}, 7(2):20:1--20:17, 2011.

\bibitem{Elkin17}
Michael Elkin.
\newblock Distributed exact shortest paths in sublinear time.
\newblock In {\em Proceedings of the 49th Annual {ACM} {SIGACT} Symposium on
  Theory of Computing, {STOC} 2017, Montreal, QC, Canada, June 19-23, 2017},
  pages 757--770, 2017.

\bibitem{ErdosR60}
P.~Erd{\H o}s and A~R{\' e}nyi.
\newblock On the evolution of random graphs.
\newblock In {\em Publication of the Mathematical Institute of the Hungarian
  Academy of Sciences}, pages 17--61, 1960.

\bibitem{EvenT75}
Shimon Even and Robert~Endre Tarjan.
\newblock Network flow and testing graph connectivity.
\newblock {\em {SIAM} J. Comput.}, 4(4):507--518, 1975.

\bibitem{Farach-ColtonHL15}
Martin Farach{-}Colton, Tsan{-}sheng Hsu, Meng Li, and Meng{-}Tsung Tsai.
\newblock Finding articulation points of large graphs in linear time.
\newblock In {\em Algorithms and Data Structures, WADS}, pages 363--372, Cham,
  2015. Springer International Publishing.

\bibitem{FeigenbaumKMSZ05b}
Joan Feigenbaum, Sampath Kannan, Andrew McGregor, Siddharth Suri, and Jian
  Zhang.
\newblock Graph distances in the streaming model: The value of space.
\newblock In {\em Proceedings of the Sixteenth Annual ACM-SIAM Symposium on
  Discrete Algorithms}, SODA '05, pages 745--754, 2005.

\bibitem{FeigenbaumKMSZ05}
Joan Feigenbaum, Sampath Kannan, Andrew McGregor, Siddharth Suri, and Jian
  Zhang.
\newblock On graph problems in a semi-streaming model.
\newblock {\em Theor. Comput. Sci.}, 348(2):207--216, December 2005.

\bibitem{FeigenbaumKSV03}
Joan Feigenbaum, Sampath Kannan, Martin~J. Strauss, and Mahesh Viswanathan.
\newblock An approximate l1-difference algorithm for massive data streams.
\newblock {\em SIAM J. Comput.}, 32(1):131--151, January 2003.

\bibitem{Fellbaum98}
Christiane Fellbaum, editor.
\newblock {\em {WordNet}: an Electronic Lexical Database}.
\newblock MIT Press, 1998.

\bibitem{FlajoletM85}
Philippe Flajolet and G.~Nigel Martin.
\newblock Probabilistic counting algorithms for data base applications.
\newblock {\em Journal of Computer and System Sciences}, 31(2):182 -- 209,
  1985.

\bibitem{FriezeK15}
Alan~M. Frieze and Michal Karonski.
\newblock {\em Introduction to Random Graphs}.
\newblock Cambridge University Press, 2015.

\bibitem{PabloL03}
Pablo~M. Gleiser and Leon Danon.
\newblock Community structure in jazz.
\newblock {\em Advances in Complex Systems}, 6(4):565--573, 2003.

\bibitem{GuhaKS01}
Sudipto Guha, Nick Koudas, and Kyuseok Shim.
\newblock Data-streams and histograms.
\newblock In {\em Proceedings of the Thirty-third Annual ACM Symposium on
  Theory of Computing}, STOC '01, pages 471--475, 2001.

\bibitem{GuruswamiK16}
Venkatesan Guruswami and Krzysztof Onak.
\newblock Superlinear lower bounds for multipass graph processing.
\newblock {\em Algorithmica}, 76(3):654--683, Nov 2016.

\bibitem{GuruswamiO16}
Venkatesan Guruswami and Krzysztof Onak.
\newblock Superlinear lower bounds for multipass graph processing.
\newblock {\em Algorithmica}, 76(3):654--683, 2016.

\bibitem{HenzingerRR98}
Monika~Rauch Henzinger, Prabhakar Raghavan, and Sridhar Rajagopalan.
\newblock Computing on data streams.
\newblock In {\em External Memory Algorithms, Proceedings of a {DIMACS}
  Workshop, New Brunswick, New Jersey, USA, May 20-22, 1998}, pages 107--118,
  1998.

\bibitem{HopcroftK73}
John~E. Hopcroft and Richard~M. Karp.
\newblock An n\({}^{\mbox{5/2}}\) algorithm for maximum matchings in bipartite
  graphs.
\newblock {\em {SIAM} J. Comput.}, 2(4):225--231, 1973.

\bibitem{HopcroftT74}
John~E. Hopcroft and Robert~Endre Tarjan.
\newblock Efficient planarity testing.
\newblock {\em J. {ACM}}, 21(4):549--568, 1974.

\bibitem{Indyk06}
Piotr Indyk.
\newblock Stable distributions, pseudorandom generators, embeddings, and data
  stream computation.
\newblock {\em J. ACM}, 53(3):307--323, May 2006.

\bibitem{KaleT17}
Sagar Kale and Sumedh Tirodkar.
\newblock Maximum matching in two, three, and a few more passes over graph
  streams.
\newblock In {\em Approximation, Randomization, and Combinatorial Optimization.
  Algorithms and Techniques, {APPROX/RANDOM} 2017, August 16-18, 2017,
  Berkeley, CA, {USA}}, pages 15:1--15:21, 2017.

\bibitem{Kapralov13}
Michael Kapralov.
\newblock Better bounds for matchings in the streaming model.
\newblock In {\em Proceedings of the Twenty-Fourth Annual {ACM-SIAM} Symposium
  on Discrete Algorithms, {SODA} 2013, New Orleans, Louisiana, USA, January
  6-8, 2013}, pages 1679--1697, 2013.

\bibitem{Khan17}
Shahbaz Khan.
\newblock Near optimal parallel algorithms for dynamic {DFS} in undirected
  graphs.
\newblock In {\em Proceedings of the 29th {ACM} Symposium on Parallelism in
  Algorithms and Architectures, {SPAA} 2017, Washington DC, USA, July 24-26,
  2017}, pages 283--292, 2017.

\bibitem{Kliemann16}
Lasse Kliemann.
\newblock Engineering a bipartite matching algorithm in the semi-streaming
  model.
\newblock In {\em Algorithm Engineering - Selected Results and Surveys}, pages
  352--378. Springer International Publishing, 2016.

\bibitem{Knuth08}
Donald~E. Knuth.
\newblock {\em The Art of Computer Programming, Volume 4, Fascicle 0:
  Introduction to Combinatorial and Boolean Functions}.
\newblock Addison-Wesley, 2008.

\bibitem{Konect13}
J{\'{e}}r{\^{o}}me Kunegis.
\newblock {KONECT - The Koblenz Network Collection.}
\newblock \url{http://konect.uni-koblenz.de/networks/}, October 2016.

\bibitem{LeskovecKF07}
Jure Leskovec, Jon Kleinberg, and Christos Faloutsos.
\newblock Graph evolution: Densification and shrinking diameters.
\newblock {\em ACM Trans. Knowledge Discovery from Data}, 1(1):1--40, 2007.

\bibitem{McAuleyJ12}
Julian McAuley and Jure Leskovec.
\newblock Learning to discover social circles in ego networks.
\newblock In {\em Advances in Neural Information Processing Systems}, pages
  548--556. 2012.

\bibitem{McGregor05}
Andrew McGregor.
\newblock Finding graph matchings in data streams.
\newblock In {\em Approximation, Randomization and Combinatorial Optimization,
  Algorithms and Techniques, 8th International Workshop on Approximation
  Algorithms for Combinatorial Optimization Problems, {APPROX} 2005 and 9th
  InternationalWorkshop on Randomization and Computation, {RANDOM} 2005,
  Berkeley, CA, USA, August 22-24, 2005, Proceedings}, pages 170--181, 2005.

\bibitem{McGregor14}
Andrew McGregor.
\newblock Graph stream algorithms: A survey.
\newblock {\em SIGMOD Rec.}, 43(1):9--20, May 2014.

\bibitem{Muthukrishnan05}
Shan Muthukrishnan.
\newblock Data streams: Algorithms and applications.
\newblock {\em Foundations and Trends in Theoretical Computer Science},
  1(2):117--236, 2005.

\bibitem{ConnellC09}
Thomas~C. O'Connell.
\newblock A survey of graph algorithms under extended streaming models of
  computation.
\newblock In {\em Fundamental Problems in Computing: Essays in Honor of
  Professor Daniel J. Rosenkrantz}, pages 455--476, 2009.

\bibitem{Ruhl03}
Jan~Matthias Ruhl.
\newblock Efficient algorithms for new computational models.
\newblock {\em PhD Thesis}, Department of Computer Science, MIT, Cambridge, MA,
  2003.

\bibitem{Tarjan72}
Robert~Endre Tarjan.
\newblock Depth-first search and linear graph algorithms.
\newblock {\em SIAM J. Comput.}, 1(2):146--160, 1972.

\bibitem{Tarjan74}
Robert~Endre Tarjan.
\newblock Finding dominators in directed graphs.
\newblock {\em {SIAM} J. Comput.}, 3(1):62--89, 1974.

\bibitem{Tarjan75}
Robert~Endre Tarjan.
\newblock Efficiency of a good but not linear set union algorithm.
\newblock {\em J. ACM}, 22(2):215--225, April 1975.

\bibitem{Tarjan76}
Robert~Endre Tarjan.
\newblock Edge-disjoint spanning trees and depth-first search.
\newblock {\em Acta Inf.}, 6:171--185, 1976.

\bibitem{TarjanL84}
Robert~Endre Tarjan and Jan van Leeuwen.
\newblock Worst-case analysis of set union algorithms.
\newblock {\em J. {ACM}}, 31(2):245--281, 1984.

\bibitem{WestbrookT92}
Jeffery Westbrook and Robert~Endre Tarjan.
\newblock Maintaining bridge-connected and biconnected components on-line.
\newblock {\em Algorithmica}, 7(5{\&}6):433--464, 1992.

\bibitem{LeskovecY12}
Jaewon Yang and Jure Leskovec.
\newblock Defining and evaluating network communities based on ground-truth.
\newblock In {\em Proc. ACM SIGKDD Workshop on Mining Data Semantics}, page~3.
  ACM, 2012.

\bibitem{ZafaraniL09}
R.~Zafarani and H.~Liu.
\newblock Social computing data repository at {ASU}, 2009.

\bibitem{Zhang10}
Jian Zhang.
\newblock A survey on streaming algorithms for massive graphs.
\newblock In {\em Managing and Mining Graph Data}, pages 393--420. Springer US,
  2010.

\end{thebibliography}

\appendix

\section{Omitted Proofs}
\label{sec:proofs}
\newtheorem*{THMdfsH}{Lemma~\ref{lem:dfsHeight}}
\begin{THMdfsH}[Min-Height Property]
Given a connected undirected graph $G$ having $m$ edges, any DFS tree of $G$ 
from any root vertex necessarily has a height $h\geq m/n$.
\end{THMdfsH}
\begin{proof}
We know that each non-tree edge in a DFS tree of an undirected graph is a {\em back edge}.
We shall associate each edge to its lower endpoint. Thus, in a DFS tree each vertex will 
be associated to a tree edge to its parent and back edges only to its ancestors. 
Now, each vertex can have only $h$ ancestors as the height of the DFS tree is $h$,
Hence each vertex has only $h$ edges associated to it resulting in less than
$nh$ edges, i.e. $m\leq nh$ or $h\geq m/n$.
Note that it is important for the graph to be connected otherwise from some root the 
corresponding component and hence its DFS tree can be much smaller.
\end{proof}

\newtheorem*{THMdfsP}{Lemma~\ref{lem:dfsProp}}
\begin{THMdfsP}
$T_C$ is a valid DFS tree of $H_C$.
\end{THMdfsP}
\begin{proof}
In order to prove this claim it is sufficient to prove that all the non-tree edges stored in $H_C$ 
are back edges in $T_C$, i.e., the endpoints of every such edge share an ancestor-descendant relationship. 
Now, invariant ${\cal I}_{T_1}$ ensures that any edge in $S_C$ having both endpoints  
in $T'_C$ is a back edge.  And invariant ${\cal I}_{T_2}$ ensures that
any edge between a vertex in $T'_C$ and $T_C\setminus T'_C$ is a back edge. Hence, all the 
non-tree edges incident on $T'_C$ (and hence all non-tree edges in $H_C$) are back edges, proving our lemma.
\end{proof}

\newtheorem*{THMhcSPACE}{Lemma~\ref{lem:HCspace}}
\begin{THMhcSPACE}
The total number of edges in $H_C$, for all the components $C$ of $G'$, are $O(nk)$. 
\end{THMhcSPACE}
\begin{proof}
The size of $H_C$ can be analysed using invariant ${\cal I}_H$ as follows.
The number of tree edges in $T_C$ (and hence in $H_C$) is $O(|V_C|)$. 
The non-tree edges stored by $H_C$ have at least one endpoint in $T'_C$. 
Using Lemma~\ref{lem:dfsProp} we know that all these edges are back edges.
To bound the number of such edges let us associate each non-tree edge to its lower endpoint.
Hence each vertex will be associated to at most $k$ non-tree edges to its $k$ ancestors in $T'_C$
(recall that $T'_C$ is the top $k$ levels of $T_C$).
Thus, $H_C$ stores $O(|V_C|)$ tree edges and $O(|V_C|\cdot k)$ non-tree edges, i.e., 
total $O(|V_C|\cdot k)$ edges. Since $\sum_{C\in G'} |V_C|\leq n$, the total number of edges 
in $H_C$ is $O(nk)$.
\end{proof}

%
%
%
%
%
%

\section{Tightness of analysis of Procedure~\ref{alg:advAlg1}}
\label{apn:advAlg1WC}
We now describe a worst case example proving the tightness of the 
analysis of Procedure~\ref{alg:advAlg1}. 
To prove this we present a sequence of edges in the input stream, 
such that in each pass the algorithm extracts the longest path of the DFS tree,
which is computed using the first $O(nk)$ edges of the input stream that are internal 
to the component. We shall show that Procedure~\ref{alg:advAlg1} will require
$\Theta(n/k)$ passes to build the DFS tree for such a graph where the height of the tree is $o(n)$.
Note that the amortized time required by the algorithm in every pass is clearly optimal upto 
$O(\alpha(m,n))$ factors, hence the focus here is merely on the number of passes.

\begin{figure}[ht]
\centering
\includegraphics[width=\linewidth]{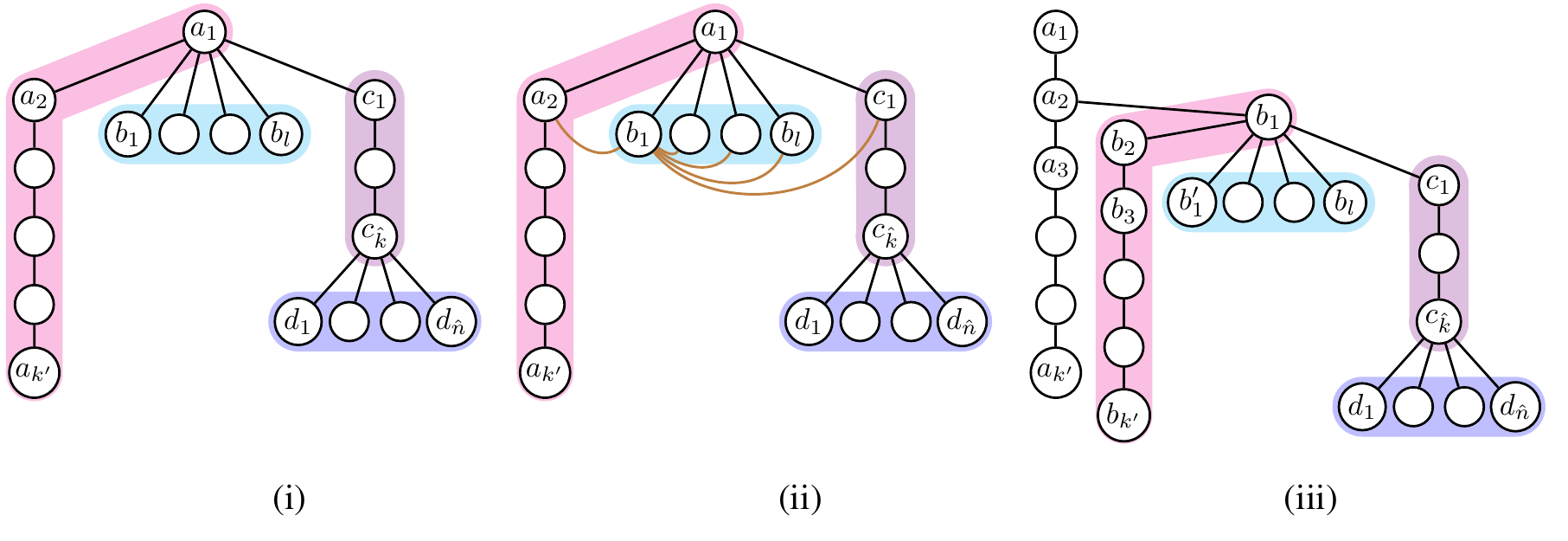}
\caption{Example to prove the tightness of analysis of Procedure~\ref{alg:advAlg1}.}
\label{fig:wcA1}
\end{figure}

Consider the tree shown in Figure~\ref{fig:wcA1} (i). At the beginning of every pass
of the algorithm, the unvisited graph will be connected and the vertices will be classified
into four sets $A,B,C$ and $D$, depending on a parameter $\hat{k}>k$ to be fixed later on. 
The set $A=\{a_1,\cdots,a_{k'}\}$ consists of $k'=\hat{k}+3$ vertices connected in the form 
of a path. The vertex $a_1$ is also the root of the DFS tree of the component.  
The set $C=\{c_1,...,c_{\hat{k}}\}$ is also connected in the form of a path. 
The set of vertices $D=\{d_1,...,d_{\hat{n}}\}$ are connected to all the vertices in $C$, 
where $\hat{n}=n/2$. The set $B=\{b_1,...,b_l\}$ contains the remaining vertices of the component 
each having an edge to $a_1$. The value of $\hat{k}$ is chosen so that the total number of edges 
described above becomes exactly the number of edges that can be stored by the procedure, i.e., 
$\Theta(nk)$. This graph of $\Theta(nk)$ edges will henceforth be called as the {\em partial graph} 
for the pass. The value of $k'=\hat{k}+3$ ensures that the longest path from root $a_1$ to a leaf, 
is the path connecting $a_1$ to the leaf $a_{k'}$. Clearly, the only DFS tree possible for the 
{\em partial} graph is shown in Figure~\ref{fig:wcA1}. Note that the subtree $T(c_1)$ may be 
computed differently with possibly one vertex of $D$ between every two vertices of $C$, 
but the corresponding analysis remains same.

Now, the $\Theta(nk)$ edges described above would appear first in the input edge sequence, and hence
would be stored by Procedure~\ref{alg:advAlg1} in the first pass, to compute the DFS tree shown 
in Figure~\ref{fig:wcA1}~(i). As a result the path connecting $a_1$ with $a_{k'}$ is added to $T$. 
The next edges in the stream connect $b_1$ to $a_2$, $c_1$ and all the vertices in $B\setminus \{b_1\}$, 
as shown in Figure~\ref{fig:wcA1}~(ii).
This makes sure 
that the unvisited graph remains connected after the end of this pass, and the lowest edge from this
component on $T$ is from $b_1$, making it the root of the DFS tree of the unvisited component. 
The corresponding spanning tree computed for the next pass is shown in Figure~\ref{fig:wcA1}~(iii). 
Notice its similarity with Figure~\ref{fig:wcA1}~(i), where $A'=\{b_1,...,b_{k'}\}$ and 
$B'=\{b'_1=b_{k'+1},...,b_l\}$. Hence, the above construction can be repeated for each pass adding a 
path of length $k'$ in every pass. Also, note that the number of edges in the {\em partial} graph 
would have now decreased as the size of $B$ has decreased (by $k'$ vertices). Hence we may have to 
add a vertex from $B$ to $C$ (by accordingly placing the corresponding edges next in the stream) to 
ensure $\Theta(nk)$ edges in the {\em partial} graph. But $C$ will grow much slower as compared to the number of passes, as each addition to $C$ adds $O(n)$ edges to the {\em partial} graph, 
whereas each pass reduces $O(k')$ edges from the {\em partial} graph. Hence the construction can continue for $\Theta(n/k)$ passes, where the size of $T$ grows by $O(1)$ in each pass resulting in a DFS tree
of height $O(n/k)$. Thus, we get the following lemma.

\begin{lemma}
There exists a sequence of edges for which Procedure~\ref{alg:advAlg1} takes $\Theta(n/k)$ passes to compute a DFS tree 
of height $o(n)$.
\end{lemma}

\begin{remark}
Showing the tightness of the number of passes without any restriction on the height of the tree is very easy. Simply order $n\choose 2$
edges of the graph by the indices of its endpoint having a lower index. Every pass would consider all the edges of $O(k)$ vertices with the lowest index, 
which will result in the final DFS tree being a single path of length $n$, computed in $O(n/k)$ passes. However, 
such an example does not highlight the importance of Procedure~\ref{alg:advAlg2} which computes a DFS tree in $O(h/k)$ passes,
where $h$ is the height of the DFS tree. 
\end{remark}


\section{Restructuring procedure~\cite{BaswanaK17}}
\label{sec:rebuild}
We now briefly describe the restructuring procedure by Baswana and Khan~\cite{BaswanaK17}.
For the sake of simplicity, here we only describe how restructuring a DFS tree $T_C$ of $H_C$
is performed on insertion of a cross edge is achieved abiding {\em monotonic fall}, i.e., 
the DFS tree is restructured such that the level of each vertex only increases. 
Hence, we will not describe the various optimizations 
used to achieve the tight bound on total update time. 
The procedure essentially adds the given cross edge into the DFS tree, restructuring it accordingly.
In particular it reverses only a single path in the tree (see Figure~\ref{fig:rerooting}). 
However, this results in some back edges of the tree to become cross edges. 
These edges can be efficiently identified and removed from the graph and inserted back 
to the graph iteratively following the same procedure. 


\begin{figure}[ht]
\centering
\includegraphics[width=.85\linewidth]{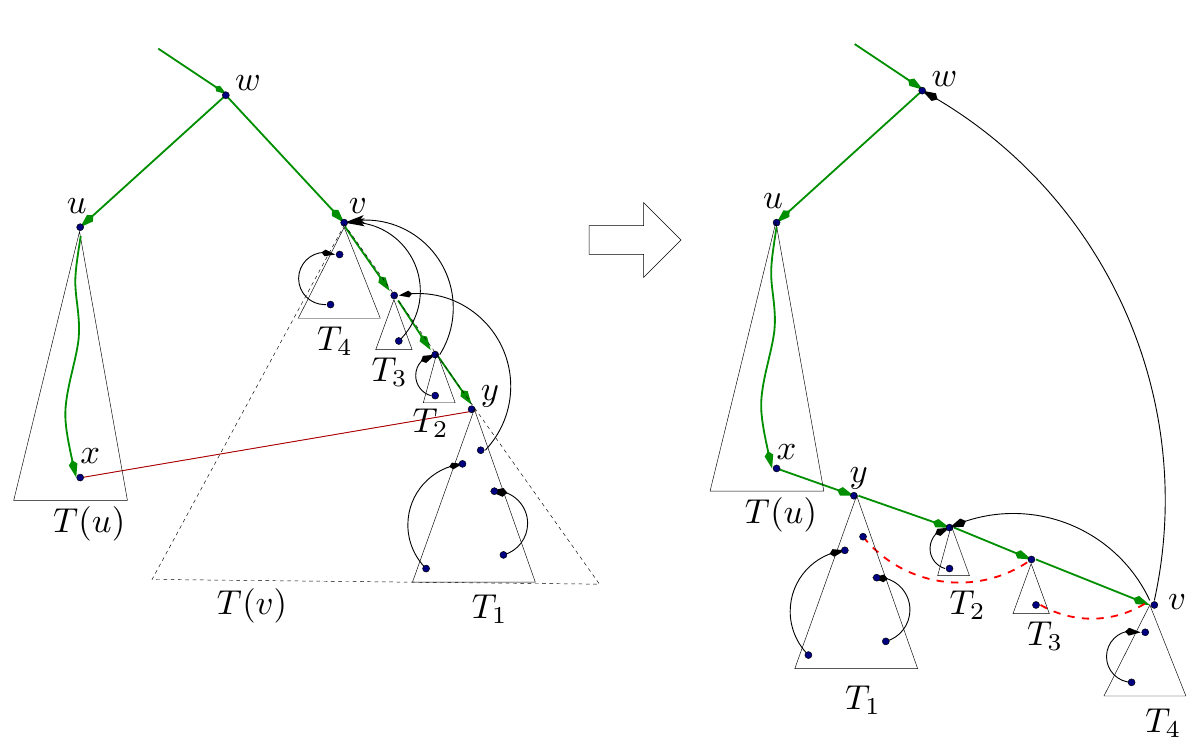}
\caption{Rerooting the tree $T_C(v)$ at $y$ and hanging it from $x$. Notice
that some back edges may become cross edges (shown dotted) due to this rerooting (reproduced from~\cite{BaswanaK17}).}
\label{fig:rerooting}
\end{figure}

Let us now describe the procedure in detail.
A pool of edges ${{\cal E}_C}$ is initialized with the inserted edge $e$. 
Then the edges in ${{\cal E}_C}$ are processed until it becomes empty. 
For each edge $(x,y)\in {{\cal E}_C}$ where $level(x) \geq level(y)$, 
it is  first removed from ${{\cal E}_C}$ and checked whether it is a cross edge.
This can be easily verified if both $x$ and $y$ are different from their
lowest common ancestor (LCA) $w$.
In case of a back edge, the edge is simply ignored  processing of ${\cal E}_C$ continues.
Otherwise, let $v$ be the ancestor of $y$ (not necessarily proper) that is a child of $w$.
Since $(x,y)$ is a cross edge, both $w$ and $v$ would surely exist.
The DFS tree $T_C$ is then restructured by removing the edge 
$(w,v)$ and adding the edge $(x,y)$. As a result the entire tree path from $y$ to $v$ 
would now hang from $y$ which was earlier  hanging from $v$,
reversing the parent child relation for all the edges on this path.
However, as a result of this restructuring several edges of $H_C$ may now become cross edges.
In order to maintain $T_C$ as the DFS tree of $H_C$, these cross edges are collected from $H_C$
and added to the pool of edges ${\cal E}_C$, which are then iteratively processed using the 
same procedure. 
Thus, the algorithm maintains the following invariant. 

\vspace{.5em}
{\centering
        \fbox{\parbox{\linewidth}{
\textbf{Invariant: } $T_C$ is a DFS tree of $H_C\setminus {\cal E}_C$.
}}}
\vspace{.5em}

\begin{procedure}
\BlankLine
${\cal E}_C\leftarrow \{e\}$\;
\While{${{\cal E}_C}\neq \emptyset$}
{
	
	$(x,y)\leftarrow$ Extract an edge from ${{\cal E}_C}$
	\tcc*{where $level(x)\geq level(y)$}
	$w\leftarrow $ LCA of $x$ and $y$ in $T_C$\;
	\BlankLine
	\If(\tcc*[f]{$(x,y)$ is a cross edge}){$w\neq x$ and $w\neq y$}
		{
		$v\leftarrow $ Ancestor of $y$ whose parent in $T_C$ is $w$\;
		Remove the edge $(w,v)$ from $T_C$\;
		Reverse the parent-child relationship of edges on path from $y$ to $v$ in $T_C$\;
		Add the edge $(x,y)$ to $T_C$\;
		
		$E_R\leftarrow$ The cross edges from $H_C\setminus {\cal E}_C$ in the tree $T_C$\;
		${{\cal E}_C}\leftarrow {{\cal E}_C}\cup E_R$\;
		}	
}
\caption{Maintain-DFS($T_C$,$e$): Maintains the DFS tree $T_C$ on insertion an edge 
$e$ abiding {\em monotonic fall} making $T_C$ a valid DFS tree of $H_C$.}
\label{alg:rebuild}
\end{procedure}

Hence, when the list ${\cal E}_C$ is empty, $T_C$ is a valid DFS tree of $H_C$. 
Refer to Procedure \ref{alg:rebuild} for the pseudocode of the restructuring procedure.
Now, observe that since $level(x)\geq level(y)$ the level of vertices can only increase as a result of the 
path reversal described above, ensuring {\em monotonic fall}. This also ensures the termination of the algorithm as the vertices cannot fall beyond the level $n$.
Moreover, the analysis~\cite{BaswanaK17} of the procedure ensures that the total work done to 
restructure $T_C$ (including maintenance of LCA structures etc.) can be associated to 
constant times the fall in level of vertices (similar to $\Phi$ described in Section~\ref{sec:alg4}). 
Since the each vertex can only fall by $h$ levels, the total fall of vertices is bounded by $O(nh)$,
where $h$ is the height of the final DFS tree. Further, recall that Procedure~\ref{alg:advAlg2} also
satisfied {\em monotonic fall}. Hence, the total time taken by Procedure~\ref{alg:rebuild} to 
restructure $T_C$ throughout the algorithm across all passes is $O(nh)$. However, we also need to 
account for the $O(1)$ time required to process an input edge whenever the procedure is invoked.
This requires total $O(n_e)$ time, where $n_e$ is the number of input edges processed by the procedure
in all the passes of the algorithm, which results in the following theorem.

\begin{theorem}
Given an undirected graph $G$, its DFS tree can be rebuild after insertion of cross edges by a procedure
abiding {\em monotonic fall} requiring total $O(nh+n_e)$ time across all invocations of the procedure, 
where $h$ is the height of the computed DFS tree and $n_e$ is the number of input edges processed. 
\end{theorem}

\begin{remark}
Baswana and Khan~\cite{BaswanaK17} showed that any algorithm maintaining a DFS tree
incrementally abiding {\em monotonic fall}, necessarily requires $\Omega(n^2)$ time 
to maintain the tree explicitly even for sparse graphs. However, if we present the bound
in terms of the height $h$ of the DFS tree, the corresponding bound reduces to $\Omega(m+nh)$
as every algorithm requires $\Omega(1)$ time to process each of the $m$ input edges. 
In the streaming environment, where multiple passes over input stream are performed, 
this bound naturally extends to $\Omega(n_e+nh)$, where $n_e$ is the number of edges processed 
during all the passes over the input stream. 
\end{remark}

%
%
 

\section{Experimental Evaluation}
\label{apn:expEval}
We now perform an experimental evaluation of the algorithms to understand their significance in practice. The main criterion of evaluation is the number of passes required to completely build the DFS tree, instead of the time taken. This makes the evaluation independent of the computing platform, programming environment, and code efficiency, resulting in easier reproduction and verification of this study. For random graphs, the results of each experiment are averaged over several test cases to get the expected behaviour.

A related experimental study was performed by Baswana et al.~\cite{BaswanaGK18} which analysed different incremental DFS algorithms on random and real graphs. For random graphs, they also presented simple single pass algorithms to build a DFS tree using $O(n\log n)$ space. 
Moreover, they also presented the following property which shall be used to describe some results during the course of our evaluation.

\begin{theorem}[DFS Height Property\cite{BaswanaGK18}]
The depth of a DFS tree of a random graph $G(n,m)$, with $m=cn\log n$ is at least $n-n/c$ with high probability.
\end{theorem}  

We shall also be using the following properties regarding the thresholds of the phase transition phenomenon of Random Graphs to describe the performance of algorithms.
\begin{theorem}[{Connectivity Threshold}\cite{FriezeK15}] 
	Graph $G(n,m)$ with $m=\frac{n}{2}(\log n + c)$  is connected with probability at 
	least $1-e^{-c}$, for any constant $c>0$.
\label{theorem:Frieze1}
\end{theorem}

\begin{theorem}[Giant Component Threshold\cite{FriezeK15}] 
	Graph $G(n,m)$ with $m=\frac{cn}{2}$ for any constant $c>0$,  the graph contains a single giant component having $O(n)$  vertices and $O(m)$ edges while the residual components have at most the size $O(\log n)$.
\label{theorem:Frieze2}
\end{theorem}

\subsection{Datasets}
In our experiments we considered the following two types of datasets.
\begin{itemize}
\item \textbf{Random Graphs:} The initial graph is the star graph, formed by adding an edge from a dummy vertex $r$ to each vertex (recall Section~\ref{sec:prelim}). The update sequence is generated based on Erd\H{o}s R\'{e}nyi $G(n, m)$ model~\cite{Bollobas84,ErdosR60} by choosing the first $m$ edges of a random permutation of all the edges in the graph. 

\item \textbf{Real graphs:} We use several publicly available undirected graphs from real world. We derived these graphs from the KONECT dataset~\cite{Konect13}. These datasets are of different types, namely, online social networks (HM~\cite{PHamster16}, Apgp~\cite{MRAA04}, BrightK~\cite{ChoML11}, LMocha~\cite{ZafaraniL09}, FlickrE~\cite{McAuleyJ12}, Gowalla~\cite{ChoML11}), human networks (AJazz~\cite{PabloL03}, ArxAP~\cite{LeskovecKF07}, Dblp~\cite{LeskovecY12}), recommendation networks (Douban~\cite{ZafaraniL09}, Amazon~\cite{LeskovecY12}), infrastructure (CU~\cite{Knuth08}, CH~\cite{Chicago1,Chicago2}),  autonomous systems (AsCaida~\cite{LeskovecKF07}), lexical words (Wordnet~\cite{Fellbaum98}) and protein base (Mpdz~\cite{ThijsLMPH05}).
\end{itemize}

\subsection{Algorithms}
During the course of these experiments, we modified the algorithms described previously using some obvious heuristics to improve their empirical performance. The analyzed algorithms are as follows.
\begin{itemize}
\item \textbf{Simple Algorithm (Simp):} This algorithm refers to Procedure~\ref{alg:simple1}  described in Section~\ref{sec:alg1}, which adds one new vertex to the DFS tree during each pass, requiring exactly $n$ passes irrespective of the data set. However, notice that after having found a new vertex $u$, the residual pass is wasted. This can be used to possibly find the next neighbour $v$ of $u$ if an edge  $(u,v)$ exists in the residual pass, and then possibly the neighbour $w$ of $v$, and so on. Hence, using this {\em additional heuristic} the algorithm now possibly adds more vertices in each pass, requiring less than $n$ passes. The algorithm without using the additional heuristic shall be referred to as \textbf{SimpO}.

\item \textbf{Improved Algorithm (Imprv):} This algorithm refers to Procedure~\ref{alg:simple} described in Section~\ref{sec:alg2}, where in the $i^{th}$ pass all the vertices in the $i^{th}$ level of the final DFS tree are added. Thus, this algorithm requires exactly $h$ passes, where $h$ is the height of the computed DFS tree. 

\item \textbf{K Path Algorithm (kPath):} This algorithm refers to Procedure~\ref{alg:advAlg1} described in Section~\ref{sec:alg3}, where by using $nk$ edges each pass adds a path of length at least $k$ to the DFS tree, for each component of the unvisited graph. For a component of $n'$ vertices, the algorithm essentially computes an {\em auxiliary} DFS tree using the first $n'k$ edges of the pass and adds the longest path of this tree to the final DFS tree. 


\item \textbf{K Level Algorithm (kLev):} This algorithm refers to Procedure~\ref{alg:advAlg2} described in Section~\ref{sec:alg4}, where in each pass by using $nk$ edges a spanning tree is computed whose top $k$ levels are the next $k$ levels of the final DFS tree. This requires {\em exactly} $h/k$ passes, where $h$ is the height of the final DFS tree. However, it is evident from the algorithm that if some vertex and all its ancestors are not modified during a pass, it will remain so in the final DFS tree. Hence, we use an {\em additional heuristic} which also adds such {\em unmodified vertices} at the end of the pass to the DFS tree. The algorithm without using the additional heuristic shall be referred to as \textbf{kLevO}.
\end{itemize}

\subsection{Experiments on Random Graphs}
\newcommand{\figW}{\linewidth}
\begin{figure*}[!ht]
\centering
\includegraphics[width=\figW]{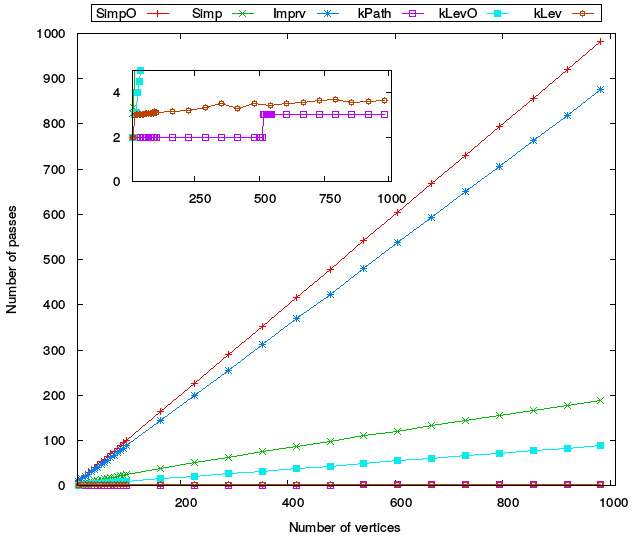}
\caption{
For various algorithms, the plot shows the number of passes required to build a DFS tree of a random graph with  $m={n \log{n}}$ edges for different values of $n$. The advanced algorithms are allowed to store $10n$ edges.  
}
\label{fig:varN}
\end{figure*}

We first compare in Figure~\ref{fig:varN}, the expected number of passes taken by different algorithms for random graphs having $n\log{n}$ edges, for different values of $n$ up to $1000$. The number of passes required by SimpO clearly matches the number of vertices as expected. The performance of Simp is strikingly better than SimpO demonstrating the significance of the {\em additional heuristic}. The variation of Imprv essentially shows the expected height of the DFS tree for $n\log n$ edges. The advanced algorithms are evaluated using $nk$ edges, for $k=10$. The algorithm kPath performs extremely well, showing the presence of deep DFS tree of a random graph even with $10n$ edges (as expected from DFS height property), and thereafter splitting into small components. It requires the minimum number of passes (recall that kPath and kLev uses an additional pass to determine the components) for the values of $n$ having $k>\log n/2$, after which it still requires merely $3$ passes. Notice that this is against the expectation because when $nk\geq m$, or $k\geq \log n$, the algorithm should require minimum passes. The number of passes taken by kLevO for a given value of $n$ is indeed close to $1/k$ times the number of passes taken by Imprv, as expected by the theoretical bounds. However, the performance of kLev is remarkably better as compared to kLevO demonstrating the significance of the {\em additional heuristic}. Apparently, the whole of DFS tree is fixed within a few passes, after which kLevO merely adds the top $k$ levels to the final DFS tree in each pass. Thus, the role of the {\em additional heuristic} is very significant, which is adversely affected as $n$ becomes larger with respect to $k$. Henceforth, we shall evaluate only Simp and kLev ignoring SimpO and kLevO as they do not seem to reveal any extra information. Following are the most surprising observations of this experiment:

\begin{observation} The advanced algorithms perform extremely well for $n$ from $1$ to $1000$,\\
(a) kPath requires the $2$ passes (minimum) until $k\leq \log n/2$ and $3$ passes henceforth.\\
(b) kLev requires merely $3$ passes which gradually increase to $4$. 
\label{obs:varN}
\end{observation}

\begin{figure*}[!ht]
\centering
\includegraphics[width=\figW]{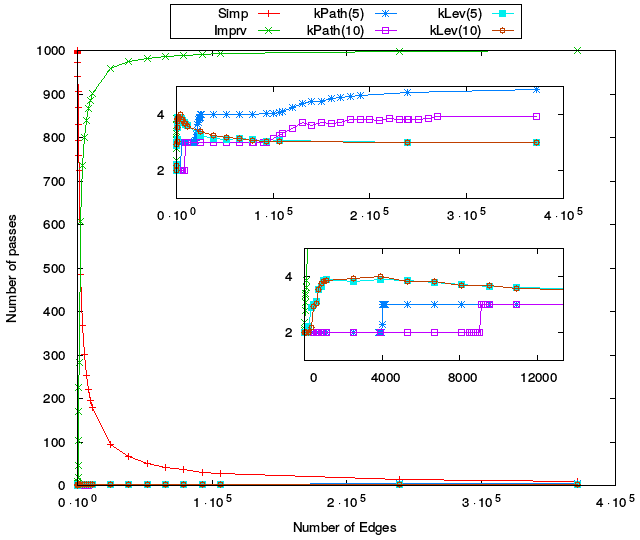}
\caption{
For various algorithms, the plot shows the number of passes required to build a DFS tree of a random graph with for $n=1000$ for different values of $m$ up to ${n \choose 2}$. The advanced algorithms are allowed to store (1) $5n$ edges, and (2) $10n$ edges.  
}
\label{fig:varM}
\end{figure*} 
 
We now compare in Figure~\ref{fig:varM}, the expected number of passes taken by different algorithms for random graphs having $1000$ vertices, for different values of $m$ up to $n \choose 2$. The number of passes required by Simp decreases sharply as the graph becomes denser. This can be explained by the {\em additional heuristic} of Simp, which has more opportunities to add vertices during a single pass in denser graphs. However, the performance of Imprv worsens sharply as the graph becomes denser. This is because the height of a DFS tree and hence the number of passes increases sharply with the density by the DFS height property. The advanced algorithms are evaluated using $nk$ edges, for $k=5$ and $10$. Notably, the performance of kPath worsens with the increase in density despite the fact that it exploits the depth of the {\em auxiliary} DFS tree to extract the longest path. However, recall that this auxiliary DFS tree is made using just the first $n'k$ edges for a component of size $n'$, which  is clearly independent of the density of the graph. Moreover, the resulting components formed after having removed the longest path would be less in number and larger in size as the density of the graph increases, justifying more number of passes required by kPath. However, notice that $k=10$ performs better than $k=5$, as $k$ clearly affects the depth of the {\em auxiliary} DFS tree and hence the length of the path added during a pass. Also, note that the passes required by kPath are minimum till little earlier than $m=nk$, as noticed in Observation~\ref{obs:varN} (a). After this, instead of increasing gradually (as its an expected value) with density as normally expected, the number of passes increases like a {\em staircase} having steps at $3$, $4$ and so on. However, after around 100,000 edges, the number of passes increases gradually for both values of $k$ to its final value, {\em not} adhering to the {\em staircase} structure. Finally, kLev also has a surprising performance sharply rising up to $4$ passes for around $1000$ edges after which it decreases gradually to settle at $3$ passes. This is despite the fact that the depth of the DFS tree is sharply increasing (see Imprv), implying that the performance is again dominated by the {\em additional heuristic}. Moreover, it also seems to be unaffected by the different values of $k$. Following are the key observations of this experiment:

\begin{observation} Performance of the advanced algorithms for varying density is as follows\\
(a) Passes required by kPath increase in steps abruptly from the $2$ (minimum) to $3$ earlier than $m=nk$, and further such transitions occur at threshold densities which increase with $k$. However, after around 100N edges, the {\em staircase} structure is not visible for both values of $k$.\\
(b) Passes required by kLev increase sharply to $4$ until $m=n$ edges, and then decrease gradually to $3$, which is independent of $k$.
\label{obs:varM}
\end{observation}

\renewcommand{\figW}{\linewidth}
\begin{figure*}[!ht]
\centering
\includegraphics[width=\figW]{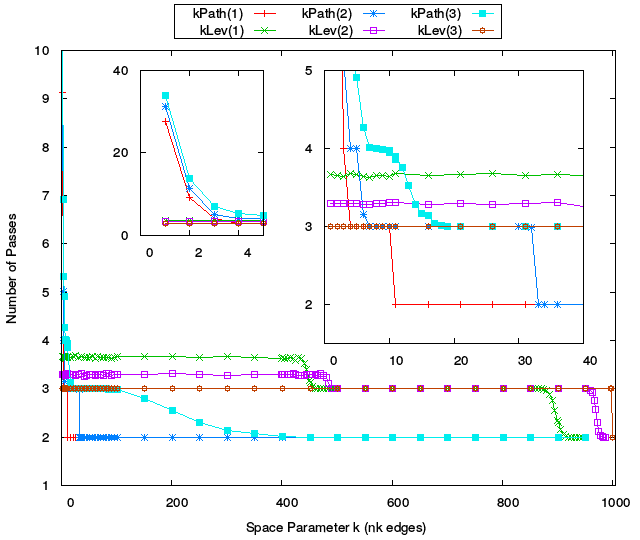}
\caption{
For various algorithms, the plot shows the number of passes required to build a DFS tree of a random graph with for $n=1000$ for different densities (1) $n\log n$, (2) $n\sqrt{n}$, and (3) $ {n \choose 2}$. The value of the space parameter $k$ (allowing local $nk$ space) is varied up to $n$. 
}
\label{fig:varS}
\end{figure*} 

Finally, we examine in Figure~\ref{fig:varS}, the variation of performance of kPath and kLev for random graphs having $1000$ vertices. We consider different values of $k$ up to $n$ and three different densities of the graph, namely, $n\log n, n\sqrt{n}$ and $n\choose 2$. This experiment thus allows us to closely examine the behaviour of the two algorithms with increasing values of $k$ as noticed in the previous experiment. The variation of the number of passes required by kPath is very interesting as it starts from a large value which sharply falls as $k$ increases. The initial value increases only slightly with the density, but the eventual fall to the minimum value is delayed with the increasing density. This sharp decline can be easily explained by the sharp increase in depth of the DFS tree of the auxiliary graph containing $nk$ edges.  However, again instead of a gradual fall, the passes decrease as a {\em staircase} with steps at $4$ and $3$ passes. Further, with the increase in density the size of a step (range of $k$ for that value) increases. Also, only in case of very high density, the descent is gradual instead of abrupt, though still following the {\em staircase} structure.  In all cases, the minimum passes are required only when $m<nk$, not when $m=nk$ as theoretically expected. The impact of density seemingly affects the size of the residual components after the longest path is removed. This adversely affects the performance of kPath. On the other hand, kLev seems to be much less affected by the variation in $k$. But still, the number of passes required decreases as a {\em staircase} in steps, with the increase in $k$. However, the increase in density clearly improves its performance where higher density reduces the passes to $3$ for smaller values of $k$. For larger values of $k$, the number passes reach the minimum value, when $k$ is larger than the expected height at the given density. Thus, the performance of kPath improves significantly with increasing $k$ and slowly with decreasing density of the graph. On the other hand, the performance of kLev (especially its {\em additional heuristic}) is affected very little by these variations, improving marginally by the increase in both $k$ and density. The key observations of this experiment are as follows.

\begin{observation} On increasing space, the advanced algorithms perform as a {\em staircase}\\
(a) Passes required by kPath decrease in steps, with threshold and size of step increasing with the increase in density. The descent between steps is abrupt, except for highest density.\\
(b) Passes required by kLev slowly decrease in steps at much higher values of $k$, which increases with the increase in density.
\label{obs:varS}
\end{observation}

\subsubsection*{Inferences from Observations}
The property of the phase transition of random graphs, seemingly explains some thresholds and the existence of steps and abrupt changes in the performance of kLev and kPath. The threshold for emergence of a giant component around $n/2$ to $n$ edges coincides with the worst performance of kLev in Observation~\ref{obs:varM}(b). This explains why prior to it when each component is very small, kLev performs better on each component individually. The threshold for connectivity seemingly explains the performance of kPath in Observation~\ref{obs:varN}(a). At this point despite $nk>m$, the presence of a much greater density in  the giant component would result in $n'k<m'$, requiring an extra pass. 

The remaining observations essentially deal with the surprising performance of kPath and the impact of the {\em additional heuristic} on kLev. Also, it demonstrates the variation of the performance in the form a {\em staircase}, with different thresholds for each step. Further, for higher densities the {\em staircase} structure is much relaxed and gradual. Among these variations we have only been able to understand the transition of kPath to minimum passes when $m<nk$. The remaining features of this exceptional performance of kPath and kLev would require further examination of the properties of random graphs, which is beyond the scope of this paper. These seem to be important future directions for understanding the behaviour of these algorithms.

\subsection{Experiments on Real Graphs}
We now evaluate the algorithms on real graphs in Table~\ref{tab:real_data}. We evaluate the advanced algorithms for different values of the space parameter, namely, $k=1,2,5$ and $10$. The evaluated real graphs are selected to cover a wide range of number of vertices and different graph densities. The density of a graph is estimated by comparing $\frac{m}{n}$ and $n$. This allows us to better predict the performance of these algorithms for even larger graphs. 

\begin{table}[!ht]
	\centering
	\resizebox{1.02\columnwidth}{!}{
	\begin{tabular}{|l|r|r|r|r|r|r|r|r|r|r|r|r|r|}
		\hline
Dataset & $n$ & $m$ & $\frac{m}{n}$ & Simp & Imprv & kPath & & & & kLev & & & \\
 & && & & & n & 2n & 5n & 10n & n & 2n & 5n & 10n \\
\hline 
CU&     49&     107&    2.18&   23&     32&     4&      4&      2&      2&      3&      3&      3&      3\\
Mpdz&   212&    242&    1.14&   181&    26&     7&      3&      2&      2&      3&      3&      3&      3\\
AJazz&  198&    2.74K&  13.85&  53&     154&    20&     14&     5&      3&      3&      3&      3&      3\\
CH&     1.47K&  1.3K&   0.88&   1.46K&  8&      5&      2&      2&      2&      2&      2&      2&      2\\
HM&     2.43K&  16.6K&  6.86&   1.31K&  753&    37&     13&     4&      2&      4&      4&      4&      4\\
Apgp&   10.7K&  24.3K&  2.28&   8.15K&  858&    35&     12&     2&      2&      4&      4&      4&      4\\
ArxAP&  18.8K&  198K&   10.55&  9.36K&  6.49K&  197&    37&     9&      4&      5&      5&      5&      5\\
AsCaida&        26.5K&  53.4K&  2.02&   24.7K&  979&    36&     8&      2&      2&      4&      4&      4&      4\\
BrightK&        58.2K&  214K&   3.68&   43.3K&  10.3K&  193&    13&     2&      2&      5&      5&      5&      5\\
LMocha& 104K&   2.19M&  21.07&  66.4K&  40K&    688&    22&     6&      4&      4&      4&      4&      4\\
FlickrE&        106K&   2.32M&  21.87&  55K&    51.7K&  599&    33&     6&      9&      5&      5&      5&      5\\
Wordnet&        146K&   657K&   4.50&   96.9K&  23.7K&  238&    51&     5&      2&      6&      6&      6&      6\\
Douban& 155K&   327K&   2.11&   145K&   11.5K&  215&    7&      2&      2&      4&      4&      4&      4\\
Gowalla&        197K&   950K&   4.83&   134K&   45.3K&  472&    30&     4&      2&      6&      6&      6&      6\\
Dblp&   317K&   1.05M&  3.31&   214K&   42.3K&  378&    35&     2&      2&      6&      6&      6&      6\\
Amazon& 335K&   926K&   2.76&   204K&   80.1K&  245&    76&     2&      2&      6&      6&      6&      6\\
\hline	
	\end{tabular}
}

\caption{Comparison of number of passes required by different algorithms on real graphs. } 
	\label{tab:real_data}
\end{table}

Overall the performance of all the algorithms follow similar pattern as in case of random graphs. The prominent high density graphs are AJazz, ArxAP, LMocha and FlickrE. The performance of Simp is close $n$ but improves with the increase in the density of the graph. However, the performance of Imprv decreases with the increase in graph density as the height of the DFS tree increases, even worse than Simp for extremely high density (see AJazz). The performance of the advance algorithms kPath and kLev is again much better than the rest even for $k=1$, except for extremely sparse graphs (see CH) where kLev and Imprv are comparable. The performance of kPath sharply improves with the value of $k$ to reach minimum value 2. The prominent exception being FlickrE where strangely the number of passes increases for $k=10$. Also, the performance is inversely affected by the increasing density, where kPath requires minimum passes even for $k=5$ when density is less. The performance of kLev again seems unaffected by the value of $k$. Also, it does not require minimum number of passes (2) even for very small graphs, except extremely sparse graph (CH) where the height of DFS tree is also very less. While for random graphs, the performance of kLev improved marginally with density, it does not seems the case with real graphs with no clear pattern emerging. Eventually for $k=10$, kPath performs marginally better than kLev for most graphs. However, for smaller values of $k$ it can perform very bad as compared to kLev.

\subsection{Results}
The two algorithms kPath and kLev perform much better than the rest even when $O(n)$ space is allowed. For both random and real graphs, kPath performs slightly worse as the density of the graph increases. On the other hand kLev performs slightly better only in random graphs with the increasing density. The effect of the space parameter is very large on kPath from $k=1$ to small constants, requiring very few passes even for $k=5$ and $k=10$. However, kLev seems to work very well even for $k=1$ and has a negligible effect of increasing the value of $k$. Overall, the results suggest using kPath if $nk$ space is allowed for $k$ being a small constant such as $5$ or $10$. However, if the space restrictions are extremely tight it is better to use kLev. 

Also, note that the superior performance of the kLev seems to be greatly attributed to the {\em additional heuristic}. Further, kPath performs much better than expected. Thus, it would be interesting to theoretically study these algorithms which seem to work extremely well in practice.

\section{Pseudocodes}
\label{sec:pseudoCode}
%

\begin{procedure}
Initialize $T\leftarrow \{r\},v\leftarrow r$\;
\While(\tcc*[f]{$T$ does not span all the vertices})
{$T$ has $<n$ vertices}
{
\For(\tcc*[f]{Initiate a pass over edges in $E$}){edges in $E$}
{
For each vertex $x\in T$, store its edge $e_x=(x,y)$ to some $y\notin T$ (if any)\;
}

\lWhile{$e_v$ is not valid edge}
{
$v\leftarrow par(v)$
}
Add $e_v$ to $T$
\tcc*[r]{say $e_v=(v,v')$, where $v\in T$ and $v'\notin T$}
$v\leftarrow v'$\;
}
\caption{Compute-DFS($r$): Computes a DFS tree of the 
component $C$ rooted at the vertex $r_c\in C$.}
\label{alg:simple1}
\end{procedure}

\begin{procedure}
Initialize $T\leftarrow \{r\}$\;
\While(\tcc*[f]{$T$ does not span all the vertices})
{$T$ has $<n$ vertices}
{
\For(\tcc*[f]{Initiate a pass over edges in $E$}){edges in $E$}
{
Compute the components $C_1,...,C_f$ of $G'$ using Union-Find algorithm\;
For each vertex $y\notin T$, store its edge $e_y$ to some leaf of $T$ (if any)\;
}

\ForEach{component $C\in C_1,...,C_f$}
{
$(x_{C},y_{C})\leftarrow$ Any valid edge $e_y$, where $y\in C$\;
Add $(x_{C},y_{C})$ to $T$\;
}
}
\caption{Compute-DFS-Improved($r$): Computes a DFS tree of the 
component $C$ rooted at the vertex $r_c\in C$.}
\label{alg:simple}
\end{procedure}

\SetKwRepeat{Do}{do}{while}
\begin{procedure}
$E'_C\leftarrow \emptyset$\;
\tcc{Initiate a pass over $E$, for all components in parallel}
\While(\tcc*[f]{process first $|V_C|k$ edges of $C$})
{$|E'_C|\leq |V_C|k$ or Stream is over}
{
\lIf{next edge $e$ belongs to $C$}
{$E'_C\leftarrow E'_C\cup \{e\}$}
}

\BlankLine
\lIf{$par(r_C)\neq \phi$}{Add $(r_C,par(r_C))$ to $T$}
$T'_C\leftarrow$ DFS tree of $T_C\cup E'_C$ from root $r_C$\;
\lIf(\tcc*[f]{pass over $E$ was completed}){$|E'_C|\leq |V_C|\cdot k$}
{Add $T'_C$ to $T$}

\Else{
$P\leftarrow$ Path from $r_C$ to lowest vertex in $T'_C$\;
Add $P$ to $T$\;
\BlankLine
\tcc{Continue the pass for all components in parallel}
\ForAll{edges in $E'_C$ followed by the remaining pass over $E$}
{
Compute the components $C_1,...,C_f$ of $C\setminus P$ using Union-Find algorithm\;
\tcc{This essentially computes $T_{C_1},...,T_{C_f}$}
Find lowest edge $e_i$ from each component $C_i$ to $P$\;
}

\ForEach{Component $C_i$ of $C\setminus P$}{
$par(y_i)\leftarrow$ $x_i$
\tcc*[r]{Let $e_i=(x_i,y_i)$, where $y_i\in C_i$}
\ref{alg:advAlg1}($C_i$,$T_{C_i}$,$y_i$)\;
}

}
\caption{Compute-DFS-Fast($C$,$T_C$,$r_C$): Computes a DFS tree of the 
component $C$ rooted at the vertex $r_C\in C$.}
\label{alg:advAlg1}
\end{procedure}

\begin{procedure}
Initialize $H_C\leftarrow T_C$\; 

\tcc{Initiate a pass over $E$, for all components in parallel}
\ForEach(\tcc*[f]{edges in the input stream from $C$})
{edge $(x,y)$ in $E$ if $(x,y)\in E_C$}
{
\lIf{$x$ and $y$ are within the same tree in $T_C\setminus T'_C$}{Continue}
		{
		Add $(x,y)$ to $H_C$\;
			\ref{alg:rebuild}($T_C,(x,y)$)\;
			Remove excess edges from $H_C$
			\tcc*[r]{non-tree edges in $T_C\setminus T'_C$}
		}
}

$T'_{C}\leftarrow$ Top $k$ levels of $T_C$
\tcc*[r]{vertices $v$ with $0\leq level\leq k-1$}
			\BlankLine
\lIf{$par(r_C)\neq \phi$}{Add $(r_C,par(r_C))$ to $T$}
Add $T'_{C}$ to $T$\;

			\BlankLine

\ForEach{tree $\tau \in T_C\setminus T'_C$}
{
$v \leftarrow root(\tau)$\;
\tcc{Let $C_v$ be component containing $v$ in $C \setminus T'_C$}
$par(v)\leftarrow$ Parent of $v$ in $T_C$\;
\ref{alg:advAlg2}($C_v$,$T_C(v)$,$v$):
}

\caption{Compute-DFS-Faster($C$,$T_C$,$r_C$): Computes a DFS tree of the 
component $C$ whose spanning tree $T_C$ is rooted at the vertex $r_C\in C$.}
\label{alg:advAlg2}
\end{procedure}

\end{document}